\documentclass{article}
\usepackage{arxiv}

\usepackage[utf8]{inputenc}
\usepackage[T1]{fontenc}
\usepackage{natbib}          
\usepackage{graphicx}
\usepackage{amsmath,amssymb,amsthm}
\usepackage{bm}
\usepackage{booktabs}
\usepackage{xcolor}
\usepackage{float}
\usepackage{tikz}
\usetikzlibrary{arrows.meta,decorations.markings}
\usepackage{url}
\usepackage{hyperref}        

\newtheorem{theorem}{Theorem}[section]
\newtheorem{proposition}[theorem]{Proposition}
\newtheorem{lemma}[theorem]{Lemma}
\newtheorem{corollary}[theorem]{Corollary}
\theoremstyle{definition}

\newtheorem{definition}[theorem]{Definition}
\theoremstyle{remark}

\newcommand{\im}{\operatorname{im}}
\newcommand{\diag}{\operatorname{diag}}

\title{Programmable Sequential Memory in Coupled Oscillator Networks}

\author{
  Taosha Guo \\
  Electrical Engineering and Computer Science\\
  University of California, Irvine\\
  \texttt{taoshag@uci.edu} \\
  \And
  Fabio Pasqualetti \\
  Electrical Engineering and Computer Science\\
  University of California, Irvine\\
  \texttt{fabiopas@uci.edu}
}

\date{}

\begin{document}
\maketitle 
\begin{abstract}
The Hopfield network established that static memories can be stored as
energy minima of an energy-based model. Real intelligent agents, however,
must navigate \emph{sequences} of memories rather than isolated
snapshots. Biological cortex addresses this through a separation of
timescales, where fast synaptic dynamics encode individual states while slow
neuromodulatory processes govern transitions between them. Inspired by
this organization, we propose a sequential memory
architecture that is fully continuous and admits exponential storage
capacity. The architecture is programmable, since the transition structure is
carried by a mechanism that is decoupled from the stored
patterns, driven by the input context, and free to be chosen,
optimized, or learned from data. Specifically, we construct an autonomous
three-timescale dynamical system with three coupled layers. A fast Kuramoto
layer stores phase patterns as exponentially stable phase-locked
configurations, an intermediate recognition layer sets the dwell times, and a
slow attention layer routes the transitions. We analyze the stability, dwell
time, and capacity of the system, and we validate the dynamics through
numerical simulations of sequential memory retrieval.  
\end{abstract} 
\section{Introduction}

Storing and retrieving \emph{static} memories with energy-based models
is by now well understood. Each pattern is encoded as a minimum of an
energy function, and recall amounts to relaxing from a cued initial
condition to the nearest minimum. In most practical tasks, however,
what an agent must remember is not an isolated snapshot but a
\emph{sequence} of states or actions, such as a trajectory to follow, a
motor program to execute, or a plan to unfold in time. \emph{Sequential}
memory is therefore the more realistic problem. Instead of settling into
one stored memory and staying there, the system must visit many stored
memories one after another, in the right order.

Recent work has addressed sequential memory directly, typically by
engineering a heteroclinic cycle, a chain of saddles, or a latching
dynamics that visits stored patterns in a prescribed order. What these constructions share is that the coupling fixes the order. The
order is encoded in an asymmetric weight matrix or in the topology of the
heteroclinic skeleton. The system can therefore replay one sequence, but it
cannot follow a different one without altering the substrate that stores the
patterns.

Biological memory suggests a way around this restriction. A recurring
theme in computational neuroscience is that memory systems separate
timescales, with a fast substrate that retrieves memories and one or
more slow variables that control it. Murray and Escola
\citeyearpar{Murray2020} identify a fast cortical pathway that supports
the rapid initial acquisition of motor skills, together with a slower
subcortical pathway through striatum and thalamus that consolidates and
protects practiced behavior. Recordings of time cells and ramping
activity in hippocampus and prefrontal cortex
\citep{Eichenbaum2014, MacDonald2011} further suggest that episodic
memory is organized by slowly varying signals that track where the
animal is in a behavioral sequence. The same idea appears in recent
associative-memory models. The Exponential Dynamic Energy Network of
Karuvally et al.\ \citeyearpar{Karuvally2025} couples a fast feature
population to a slow modulator with asymmetric interactions, which
gives stable transitions between memories, closed-form escape times, and
exponential sequence capacity. The input-driven plasticity model of
Betteti et al.\ \citeyearpar{Betteti2026} pairs fast associative
retrieval with a slow variable whose gain gradually destabilizes the
current attractor. Asymmetric modern Hopfield networks
\citep{Chaudhry2023} reach a sequence capacity that grows exponentially
with network size by combining higher-order interactions with temporally
asymmetric coupling. Across these different systems the design is the
same. A fast layer holds the memories, and a slow layer drives the
transitions between them. In these models, however, the order of the
transitions is still encoded together with the stored memories. In this
paper we let the slow layer carry the order itself, so that the order
can be changed without touching the stored memories.

\subsection*{Main contributions}

The main contribution of this paper is a smooth, autonomous,
three-timescale dynamical system built on Kuramoto oscillator networks that
separates \emph{what} is stored from \emph{in what order} it is visited.
The fast Kuramoto layer stores a set of $M$ continuous
phase-difference patterns as exponentially stable phase-locked configurations, each
realized by an exact closed-form weight vector. A slow attention vector
$\alpha\in\Delta^M$ governs transitions: its replicator dynamics, driven
by a fitness signal $\beta K^\top r$ where $r$ is a lagged recognition
vector, steer $\alpha$ along the simplex so that the closed-loop
trajectory visits the stored patterns in the order encoded by the routing
matrix $K$. An intermediate recognition layer sets the dwell
time at each memory. We analyze each layer of the system, covering stability, transition timing,
and capacity. Numerical experiments test the dwell-time prediction and the
context-dependent routing (Sections~\ref{sec:setting}--\ref{sec:experiments}).

Our system differs from existing sequential-memory models in three
ways. First, earlier constructions such as heteroclinic cycles, latching
networks, and asymmetric Hopfield networks encode the order of recall in
the coupling weights, so replaying a different sequence requires
rewriting the weights that store the memories. Here the order is carried
by the routing matrix $K$, which is separate from those weights.
Changing $K$ sends the system through the same stored memories in a
different order without altering a single weight, and since $K$ is a
free parameter, it can be chosen, optimized, or learned from data.
Second, classical Hopfield networks restrict memories to binary configurations
in $\{\pm 1\}$, whereas our Kuramoto network stores continuous
phase-difference patterns, which allows a larger class of memories. Third, the order of recall is not tied to a single sequence.
An external context signal $u(t)\in\Delta^{S}$ selects among a set of
routing matrices $K^{(0)},\dots,K^{(S-1)}$ through the combination
$K(u)=\sum_{s}u_{s}K^{(s)}$, so the same stored memories are recalled in
different orders depending on the context, with no change of weights and
no reinitialization of the network. The stored memories, the oscillator
dynamics, and the order of recall are therefore independent of one
another, and the order can be switched while the system is running.

\subsection*{Related work}

The study of sequential memory in neural networks begins with the
observation that standard associative-memory models, from classical
Hopfield networks \citep{Sompolinsky1986, Kleinfeld1986} to their
modern high-capacity successors \citep{Krotov2016, Demircigil2017,
Ramsauer2021}, are static. They store patterns as
isolated energy minima and have no natural mechanism for moving
between them in a prescribed order. Early attempts to introduce
sequence extended these models with temporally asymmetric Hebbian
rules, hard-wiring a traversal order into the coupling weights at the
cost of sublinear capacity and an inability to reprogram the sequence
without rewriting the weights. More recent constructions achieve
exponential sequence capacity through asymmetric higher-order
interactions \citep{Chaudhry2023}, but the same fundamental limitation
remains. The route is inseparable from the storage substrate.

A parallel line of work, motivated by dynamical-systems theory rather
than energy landscapes, builds sequence generators out of
heteroclinic cycles and winnerless-competition networks
\citep{Rabinovich2008, Russo2012, Recanatesi2015}. In these models a
trajectory is guided from one saddle point to the next by carefully
engineered unstable manifolds, and the rigorous stability theory
developed by Krupa--Melbourne \citeyearpar{Krupa1995}, Field
\citeyearpar{Field1996}, and Ashwin--Postlethwaite \citeyearpar{Ashwin2013}
provides precise conditions under which such cycles persist. Our
system is related to these models, but we do not engineer a
heteroclinic skeleton by hand. Instead, the programmable routing matrix $K$
sets the order, which separates the question of \emph{what} is stored from the
question of \emph{in what order} it is visited.

The choice of Kuramoto oscillators as the storage substrate connects
our work to a third body of literature. Phase-locking theory for
Kuramoto networks on graphs is developed systematically by
D\"{o}rfler and Bullo \citeyearpar{Dorfler2014}, and the splay-state
stability results that underpin our memory-as-phase-locked configuration
construction \eqref{eq: weights mapping} follow from this framework. To our
knowledge, the inverse-design step, which derives coupling weights that
make a prescribed phase pattern exponentially stable, has not
previously been applied to programmable sequential memory.

Finally, and most directly relevant to our architecture, a growing
body of work in both neuroscience and computational modeling argues
that memory and cognition are organized across multiple timescales
rather than a single one. Experimental evidence from motor learning
identifies fast cortical and slow subcortical pathways operating in
parallel \citep{Murray2020}, while time-cell and ramping-activity
recordings in hippocampus and prefrontal cortex suggest that episodic
memory is indexed by slowly varying signals that track position within
a behavioral sequence \citep{Eichenbaum2014, MacDonald2011}. This
biological intuition has recently been formalized in two-timescale
Hopfield variants. The dynamic energy networks of Karuvally et al.\
\citeyearpar{Karuvally2025} and the input-driven plasticity networks of
Betteti et al.\ \citeyearpar{Betteti2026} both couple a fast retrieval layer
to a slow modulatory variable that progressively destabilizes the
current attractor, yielding closed-form transition times and
provably stable sequence retrieval. Our three-timescale system
extends this program by adding a third, intermediate layer that
regulates dwell times, replacing the Hopfield energy substrate with a
Kuramoto phase network, and making the routing rule an explicit,
interchangeable input rather than a consequence of the weights.

\section{Encoding Memory Sequences through Phase Synchronization}\label{sec:setting}
 Consider an undirected network $\mathcal G = (\mathcal{O}, \mathcal{E})$ of $n$ oscillators with dynamics
\begin{align}{\label{eq: system dynamics}}
\dot{\theta}_i = \omega_i +\sum_{j=1}^n a_{ij} \sin (\theta_j - \theta_i),
\end{align}
where $\omega_i \in \mathbb{R}$ and $\theta_i \in \mathbb{S}^1$ are the frequency and phase of the $i$th oscillator, respectively, $a_{ij}>0$ is the coupling weight between oscillators $i$ and $j$, and $\mathcal{O} = \{1, \dots,n\}$ and $\mathcal{E} \subseteq \mathcal{O} \times \mathcal{O}$ with $|\mathcal{E}| = N$ denote the oscillator and interconnection sets, respectively. In this manuscript, we use the Kuramoto oscillator network \eqref{eq: system dynamics} as the substrate for sequential memory. Specifically, we store each memory pattern as a stable phase-locked configuration, equivalently an energy minimum of the dynamics. Our goal is to design a mechanism that automatically switches the coupling weights $a_{ij}$ to reshape the energy landscape of \eqref{eq: system dynamics} so that each successive pattern becomes the new minimum and the network flows continuously from one memory to the next.

We first show how memory patterns can be stored as stable phase-locked configurations of the Kuramoto network.
Let $\delta_e = \theta_j - \theta_i$ be the phase difference across edge $e=(i,j)$, and let
$\boldsymbol{\delta} \in \mathbb{R}^{N}$ be the vector of all phase differences.  Throughout the paper, memories are represented as phase differences $\boldsymbol{\delta}$, and we set the natural frequencies $\omega_i = 0$ for simplicity. Let $B\in\mathbb{R}^{n\times N}$ be the incidence matrix of $\mathcal{G}$, so that $\boldsymbol{\delta} = B^{\top} \boldsymbol{\theta}$ with $\boldsymbol{\theta}\in\mathbb{R}^{n}$ the vector of all phases. The next lemma states when $\boldsymbol{\delta}$ is a stable phase-locked configuration.

\begin{lemma}[Stable phase-locked configurations \citep{Menara2022}]\label{lem:stable-config}
A configuration $\boldsymbol{\delta} = B^{\top}\boldsymbol{\theta}$ is a stable phase-locked configuration (equilibrium) of \eqref{eq: system dynamics} if
\begin{enumerate}
\item[(i)] $B\,\diag(\boldsymbol{a})\sin(\boldsymbol{\delta}) = 0$; and
\item[(ii)] $\boldsymbol{\delta} \in (-\pi/2, \pi/2)^{N}$.
\end{enumerate}
\end{lemma}

\noindent Condition~(i) is the equilibrium condition $\dot{\boldsymbol{\theta}} = 0$, and condition~(ii) guarantees that the equilibrium is stable. 

Next we characterize the condition under which the system can successfully move from one memory pattern to the next memory pattern sequentially. Suppose the system rests at the current pattern $\boldsymbol{\delta}^k$, a stable equilibrium under the
weights $a_{ij}^k$, and we want to move it to a new pattern $\boldsymbol{\delta}^{k+1}$ by switching the
weights from $a_{ij}^k$ to $a_{ij}^{k+1}$. For this transition to
succeed, $\boldsymbol{\delta}^{k+1}$ must be a stable equilibrium of the new dynamics and the old memory pattern
$\boldsymbol{\delta}^k$ must lie in its basin of attraction, so that the dynamics of \eqref{eq: system dynamics} under $a_{ij}^{k+1}$
carries $\boldsymbol{\delta}^k$ to $\boldsymbol{\delta}^{k+1}$ (see Figure~\ref{fig:transition1}). We summarize the specific conditions in the following theorem.

\definecolor{memgold}{RGB}{221,163,0}
\definecolor{memnavy}{RGB}{22,32,64}
\newcommand{\memwell}[3]{%
  \draw[#3] (#1,#2) circle (0.24);
  \draw[#3] (#1,#2) circle (0.45);
  \draw[#3] (#1,#2) circle (0.66);
}

\begin{figure}[ht]
\centering
\begin{tikzpicture}[>={Stealth[length=5pt]},line cap=round]
  \def\Ax{1.05} \def\Ay{2.80} 
  \def\Bx{2.95} \def\By{1.45} 
  \begin{scope}
    \draw[rounded corners=7pt,gray!55] (0,0) rectangle (4.2,3.8);
    \memwell{\Ax}{\Ay}{gray!55}
    \fill[memgold] (\Ax,\Ay) circle (3.4pt);
    \fill[memnavy] (\Ax+0.14,\Ay-0.10) circle (1.7pt);
    \node[memnavy,font=\small] at (\Ax+0.58,\Ay-0.16) {$\delta(t)$};
    \fill[memnavy] (\Bx,\By) circle (2.2pt);
    \node[memnavy,font=\small] at (\Ax-0.62,\Ay+0.82) {$\delta^k$};
    \node[memnavy,font=\small,right=1pt] at (\Bx,\By+0.32) {$\delta^{k+1}$};
    \node[font=\small\bfseries] at (2.1,-0.55) {(a)};
  \end{scope}
  \begin{scope}[xshift=5.2cm]
    \draw[rounded corners=7pt,gray!55] (0,0) rectangle (4.2,3.8);
    \draw[memgold!65,fill=memgold!12,thick]
      plot[smooth cycle,tension=0.85] coordinates
      {(0.5,2.55) (1.2,3.4) (2.7,2.95) (3.95,2.05) (4.05,0.95) (3.05,0.35) (1.7,0.5) (0.55,1.6)};
    \node[memgold!75!black,font=\small] at (1.45,1.15) {basin of $\delta^{k+1}$};
    \memwell{\Bx}{\By}{gray!55}
    \fill[memgold] (\Bx,\By) circle (3.4pt);
    \fill[memnavy] (\Ax,\Ay) circle (2.0pt);
    \node[memnavy,font=\small] at (\Ax-0.48,\Ay+0.58) {$\delta^k$};
    \draw[memnavy,very thick,->] (\Ax+0.24,\Ay-0.16)
      to[bend left=10] (\Bx-0.58,\By+0.52);
    \node[gray,font=\small] at (2.1,3.55) {$a^k\!\to\!a^{k+1}$};
    \node[font=\small\bfseries] at (2.1,-0.55) {(b)};
  \end{scope}
\end{tikzpicture}
\caption{Retrieval by weight switching (Theorem~\ref{thm: general graph}). Memories are points of the phase space $\delta$, and the weights decide which one is stable. Panel~(a) shows that the current weights $\boldsymbol{a}^k$ carve an energy well around $\delta^k$, and the state locks there. Panel~(b) shows that switching the weights $\boldsymbol{a}^k\!\to\!\boldsymbol{a}^{k+1}$ moves the well to the next memory $\delta^{k+1}$. Because $\delta^k$ lies in the basin of $\delta^{k+1}$ (conditions~(ii)--(iii)), the state follows the well and settles at $\delta^{k+1}$.}
\label{fig:transition1}
\end{figure}

\begin{theorem}[Retrieval by weight switching]\label{thm: general graph}
Let $\boldsymbol{\delta}^k$ be the current pattern, i.e.\ a stable equilibrium of
\eqref{eq: system dynamics} under the weights $\boldsymbol{a}^k$, and let $\boldsymbol{\delta}^{k+1}$ be the pattern we
want the network to transition to after switching the weights from $\boldsymbol{a}^k$ to $\boldsymbol{a}^{k+1}$. Let $E^{k+1}$ be the energy function under the new weights, $\mathcal R^{k+1}$ the set of configurations reachable from
$\boldsymbol{\delta}^{k+1}$, and $E^{k+1}_\partial$ the lowest energy on its boundary, respectively:
\[
  E^{k+1}(\boldsymbol{\delta})=-\sum_{e\in\mathcal E} a^{k+1}_e\cos\delta_e,
  \qquad
  \mathcal R^{k+1}=\bigl(\boldsymbol{\delta}^{k+1}+\im B^\top\bigr)\cap\Bigl(-\tfrac\pi2,\tfrac\pi2\Bigr)^{N},
  \qquad
  E^{k+1}_\partial=\inf_{\boldsymbol{\delta}\in\partial\mathcal R^{k+1}}E^{k+1}(\boldsymbol{\delta}).
\]
Then, after the switch, the dynamics of \eqref{eq: system dynamics} started from $\boldsymbol{\delta}^k$
converges to $\boldsymbol{\delta}^{k+1}$ if the following conditions are satisfied
\begin{enumerate}
\item[(i)] $\boldsymbol{\delta}^{k+1}\in(-\tfrac\pi2,\tfrac\pi2)^{N}$ and
      $\diag(\boldsymbol{a}^{k+1})\sin(\boldsymbol{\delta}^{k+1})\in\ker B$;
\item[(ii)] $\boldsymbol{\delta}^{k}\in\mathcal R^{k+1}$; 
\item[(iii)] $E^{k+1}(\boldsymbol{\delta}^k)<E^{k+1}_\partial$.
\end{enumerate}
\end{theorem}

Condition (i) says the new weights $\boldsymbol{a}^{k+1}$ make
$\boldsymbol{\delta}^{k+1}$ a stable equilibrium. Condition (ii) says the old
pattern $\boldsymbol{\delta}^k$ can actually be reached from $\boldsymbol{\delta}^{k+1}$ by the dynamics:
since $\boldsymbol{\delta}=B^\top\boldsymbol{\theta}$, the flow moves $\boldsymbol{\delta}$ only within $\im B^\top$,
and on the resulting reachable set $\mathcal R^{k+1}$ the well at $\boldsymbol{\delta}^{k+1}$ is the only
equilibrium, so the state cannot be captured by any other memory. Condition (iii) says
$\boldsymbol{\delta}^k$ starts below the lowest energy on the boundary of $\mathcal R^{k+1}$. Because the
energy only decreases along the flow, the state can never climb out of $\mathcal R^{k+1}$ and instead settles
at $\boldsymbol{\delta}^{k+1}$. Together, (ii) and (iii) are exactly the statement that
$\boldsymbol{\delta}^k$ lies in the basin of attraction of $\boldsymbol{\delta}^{k+1}$.

The analysis above tells us what a substrate for sequential memory needs: patterns $\{\boldsymbol{\delta}^1,\boldsymbol{\delta}^2,\dots,\boldsymbol{\delta}^M\}$ that are stable equilibria by condition~(i), and weights that route the state from one to the next by conditions~(ii)--(iii) of Theorem~\ref{thm: general graph}. We now give a specific choice of network that meets these conditions and maps weights to phase-locked configurations easily, namely the honeycomb graph \cite{Ogranovich2026}, which is given by the following definition.
\begin{definition}[Honeycomb graph]\label{def:honeycomb}
  For integers $m \geq 1$ and $n_c \geq 5$, the \emph{honeycomb graph}
  $\mathcal{G}_m^{n_c}$ is constructed as follows:
  \begin{enumerate}
  \item Create $m$ cycles, denoted
    $\mathcal{C}_1, \mathcal{C}_2, \ldots, \mathcal{C}_m$, where
    each cycle initially contains $n_c$ nodes.
  \item For each $\ell \in \{1, \ldots, m-1\}$, merge one node from
    cycle $\mathcal{C}_\ell$ with one node from cycle $\mathcal{C}_{\ell+1}$ 
    into a single shared node.
  \end{enumerate}
  The resulting network has $n = m(n_c - 1) + 1$ nodes and $N = m\,n_c$ edges, and forms a
  one-dimensional chain of interlocking cycles. 
\end{definition}
There are many ways to number the nodes and pick which ones the cycles
share. Throughout this paper we follow the same structure in \cite{Ogranovich2026}, where cycle
$\mathcal{C}_\ell$ uses nodes $(\ell-1)(n_c-1) + 1$ to $\ell(n_c-1) + 1$, so its
last node is the first node of the next cycle $\mathcal{C}_{\ell+1}$.
We choose the honeycomb graph because, once a target pattern satisfies the hypotheses of Theorem~\ref{thm: general graph}, the weights that realize it can be written down in closed form. The following lemma makes this construction precise.
\begin{lemma}[Weight design on the honeycomb graph]
\label{lemma: weights mapping}
Let $\mathcal{G}_m^{n_c}$ be a honeycomb graph defined in Definition~\ref{def:honeycomb}, and let $\boldsymbol{\delta}^{k+1}\in\mathbb{R}^N$ be a target pattern such that
\begin{enumerate}
  \item[(i)] $\boldsymbol{\delta}^k \in  \mathcal{R}^{k+1} $;
  \item[(ii)] $\boldsymbol{\delta}^{k+1}\in\bigl(-\tfrac{\pi}{2},\tfrac{\pi}{2}\bigr)^N$ with $\delta^{k+1}_{ij}\neq 0$ on every edge; and
  \item[(iii)] the phase advances monotonically around each cycle $\mathcal{C}_\ell$. That is,
  \[
    \operatorname{sign}\bigl(\delta^{k+1}_{ij}\bigr)=s_\ell
    \qquad\text{for every edge }(i,j)\in\mathcal{C}_\ell .
  \]
\end{enumerate}
Then the coupling weights
\begin{align}\label{eq: weights mapping}
  a^{k+1}_{ij} \;=\; \frac{1}{\bigl|\sin \delta^{k+1}_{ij}\bigr|}
\end{align}
make $\boldsymbol{\delta}^{k+1}$ a stable phase-locked configuration of \eqref{eq: system dynamics}. Equivalently, $\boldsymbol{a}^{k+1}$ satisfies condition~(i) of Theorem~\ref{thm: general graph}.
\end{lemma}

Lemma~\ref{lemma: weights mapping} gives a simple formula for the weights that carry the system from one pattern to the next. The formula \eqref{eq: weights mapping} works because the honeycomb graph has a special kernel structure in its incidence matrix $B$. We refer readers to the proof of Lemma~\ref{lemma: weights mapping} in the Appendix for more details. As an example, Figure~\ref{fig:honeycomb_encoding}(a) shows the honeycomb graph $\mathcal{G}_3^5$, consisting of three cycles $\mathcal{C}_1$, $\mathcal{C}_2$, and $\mathcal{C}_3$ with five nodes each. Panel~(b) lists the memory pattern $\boldsymbol{\delta}$ stored as a stable phase-locked configuration on this honeycomb graph, together with the corresponding edge weights given by \eqref{eq: weights mapping}.

This section introduced an oscillator-based substrate in which patterns are stored as stable phase-locked configurations. On the honeycomb graph the network weights have a closed-form design, so the network can be moved from one memory to the next. The next section builds on this substrate a mechanism that lets the system traverse the memory sequence autonomously, without switching the weights by hand.

\begin{figure}[ht]
\centering
\includegraphics[width=\textwidth]{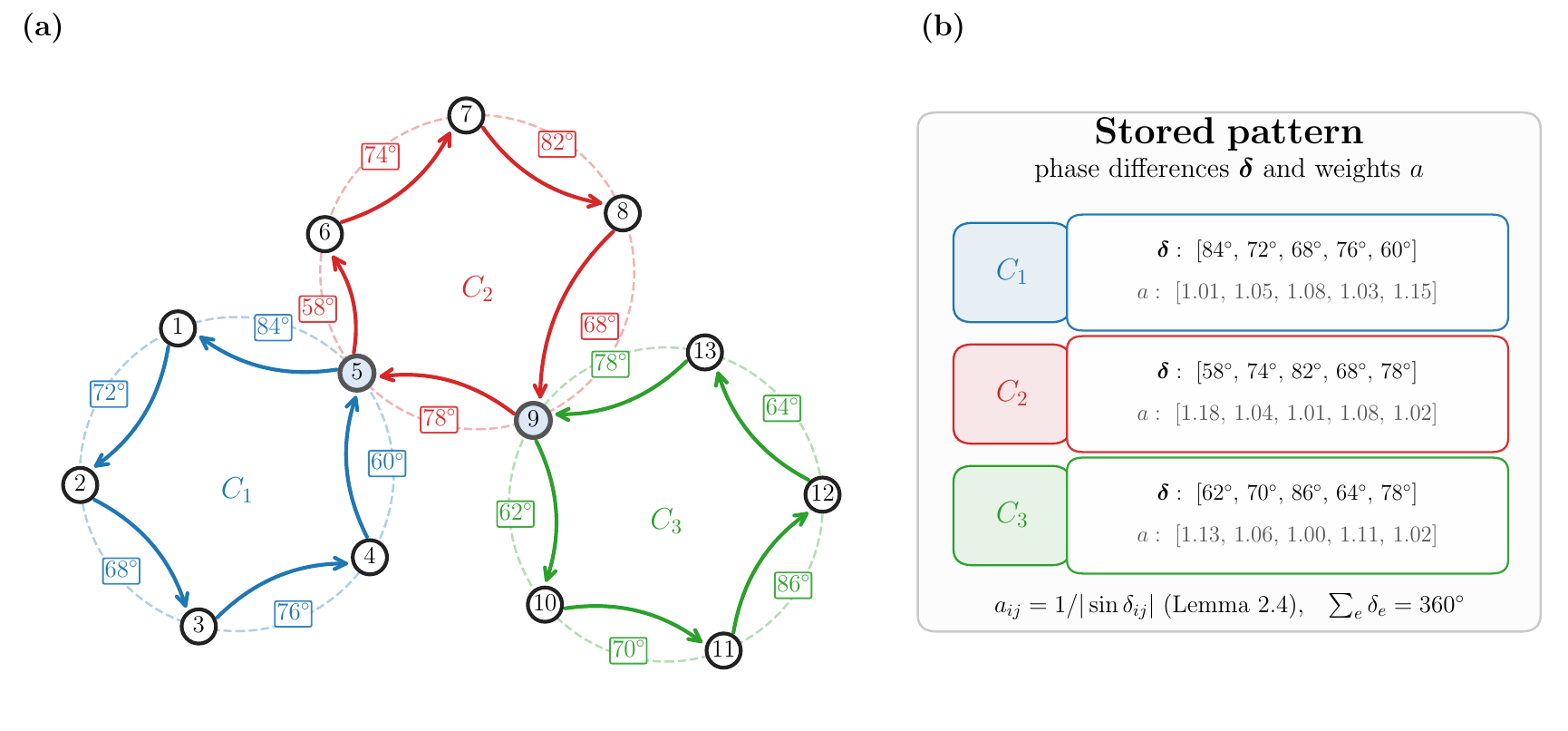}
\caption{Storing a pattern on the honeycomb $\mathcal{G}_3^5$. Panel~(a) shows three $5$-cycles $\mathcal{C}_1, \mathcal{C}_2, \mathcal{C}_3$ that share the nodes $5$ and $9$. Each edge is labeled by its phase difference $\delta_{ij}$, and the values around each cycle sum to $360^\circ$ (winding~$1$). Panel~(b) shows the stored pattern $\boldsymbol{\delta}$ together with the weights $a_{ij}=1/|\sin\delta_{ij}|$ that Lemma~\ref{lemma: weights mapping} assigns to make $\boldsymbol{\delta}$ a stable phase-locked configuration.}
\label{fig:honeycomb_encoding}
\end{figure}
\newpage

\section{A three-timescale system for autonomous sequential recall}\label{sec:architecture}

We now introduce a three-timescale dynamical system that lets the honeycomb network $\mathcal{G}_m^{n_c}$ traverse a sequence of stored memory patterns autonomously. Let $\boldsymbol{\delta}^1,\dots,\boldsymbol{\delta}^M$ be the memory patterns, each satisfying  Lemma~\ref{lemma: weights mapping}, and let $\boldsymbol{a}^1,\dots,\boldsymbol{a}^M$ be the corresponding weights given by \eqref{eq: weights mapping}. Let also $\Delta^M$ denote the set of probability vectors in $\mathbb{R}^M$, whose vertices $b_1,\dots,b_M$ are the standard basis vectors. The system has one state variable on each timescale. On the fast timescale, the oscillator phases $\boldsymbol{\theta}\in\mathbb{R}^n$ settle to the stored pattern whose weights are currently active. On the intermediate timescale, a recognition vector $r\in\Delta^M$ maintains a lagged estimate of which stored pattern the network has been visiting. On the slow timescale, an attention vector $\alpha\in\Delta^M$ selects which weights $\boldsymbol{a}^k$ are currently active. The joint dynamics are
\begin{equation}\label{eq:joint}
\begin{aligned}
\tau_{\theta}\dot\theta_i &= \sum_{j=1}^{n} a_{ij}(\alpha)\sin(\theta_j-\theta_i),\\[2pt]
\tau_r\,\dot r &= \widetilde q(\boldsymbol{\theta}) - r,\\[2pt]
\tau_{\alpha} \dot\alpha_k &= \,g(\boldsymbol{\theta})\,\alpha_k\,
\beta\bigl((K^\top r)_k - \langle K^\top r,\alpha\rangle\bigr)
\end{aligned}
\end{equation}
where
\begin{align*}
a_{ij}(\alpha)&=\sum_{k=1}^{M}\alpha_k\,a^k_{ij},\\
\widetilde q_j(\boldsymbol{\theta})&=
\frac{\exp\!\bigl(-\gamma\|\boldsymbol{\delta}^j-\boldsymbol{\delta}\|^2\bigr)}
     {\sum_{k=1}^{M}\exp\!\bigl(-\gamma\|\boldsymbol{\delta}^k-\boldsymbol{\delta}\|^2\bigr)},\qquad j=1,\dots,M,\\
g(\boldsymbol{\theta})&=\exp\!\bigl(-n\,\mathrm{Var}(\dot{\boldsymbol{\theta}})/\delta_g\bigr).
\end{align*}
Here $\widetilde q$ is a softmax with sharpness $\gamma>0$ over the distances between the current phase differences $\boldsymbol{\delta}$ (taken modulo $2\pi$) and the stored patterns, and $g$ is a convergence gate with threshold $\delta_g>0$, which is open ($g=1$) when the network settles at any memory patterns, namely when  $\dot{\boldsymbol{\theta}}=0$. The matrix $K\in\mathbb{R}^{M\times M}$ is a row-stochastic routing matrix whose row $k$ specifies which memory follows memory $k$, $\beta>0$ is a selection strength, and $\tau_\theta,\tau_r,\tau_\alpha>0$ are time constants. 

The three state variables evolve on separated timescales, $\tau_\theta<\tau_r<\tau_\alpha$, and we explain them from fastest to slowest. On the fast timescale $\tau_\theta$, the attention $\alpha$ is effectively frozen and the phases $\boldsymbol{\theta}$ converge to a phase-locked configuration under the weights $a_{ij}(\alpha)$. At a vertex $\alpha=b_k$ these weights reduce to $\boldsymbol{a}^k$, so by Lemma~\ref{lemma: weights mapping} the network settles at the stored pattern $\boldsymbol{\delta}^k$. 

On the intermediate timescale $\tau_r$, the recognition vector $r$ records which memory the network is currently at. The softmax $\widetilde q$ is close to $b_k$ whenever $\boldsymbol{\delta}$ is close to $\boldsymbol{\delta}^k$, and $r$ tracks $\widetilde q$ with time constant $\tau_r$. Thus $r$ is a lagged copy of $\widetilde q$, which means it reaches $b_k$ only after the phases have stayed at $\boldsymbol{\delta}^k$ for some time, and it retains the previous memory for about as long after they leave. 

Finally, on the slow timescale $\tau_\alpha$, the attention $\alpha$ selects the next memory. Which memory comes next is set by the routing matrix $K$ and $r$, that is, if the network is at memory $k$, then $r\approx b_k$, and $\alpha$ moves toward the memory that follows $k$ in $K$. When $\alpha$ moves is set by the gate, which is determined by whether $\boldsymbol{\theta}$ is at a steady state or is moving between memories. $g\approx1$ once the phases have settled at a stored pattern, and $g\approx0$ while they are moving between patterns.  

In summary, the three timescales divide the work. On the fast timescale the phases $\boldsymbol{\theta}$ settle at a stored pattern. On the intermediate timescale the recognition vector $r$ detects which pattern the network is at and keeps the network there for a time of order $\tau_r$. On the slow timescale the attention $\alpha$ uses $K$ and $r$ to move the network to the next pattern, while the gate $g$ keeps $\alpha$ frozen until the phases have locked at the current pattern. Repeating this cycle retrieves the stored patterns in the order specified by $K$, without switching the weights by hand. Figure~\ref{fig:orbit} illustrates one such cycle through four stored patterns.
\begin{figure}[!htbp]
\centering
\includegraphics[width=0.89\linewidth]{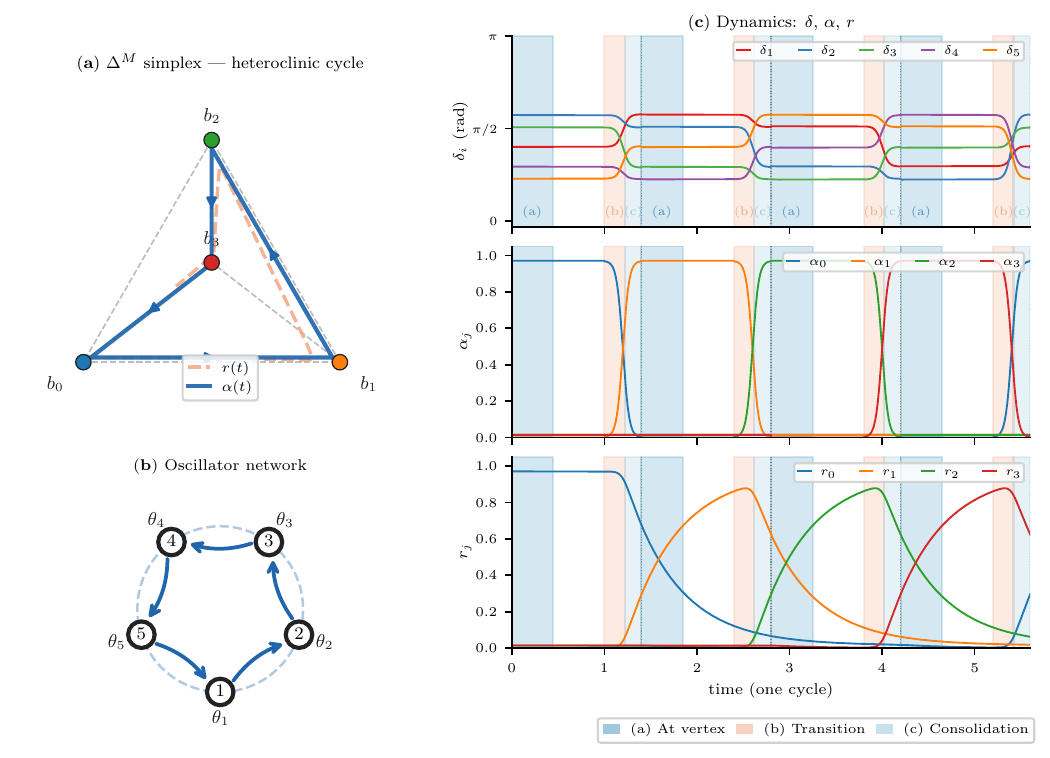}
\caption{\label{fig:orbit}One cycle of the three-timescale system \eqref{eq:joint} through four stored patterns.
(a) With $\alpha=b_k$, the phases $\boldsymbol{\theta}$ settle at
$\boldsymbol{\delta}^k$, the gate is open ($g\approx1$), and $r$ approaches
$b_k$ on the timescale $\tau_r$. The attention $\alpha$ then starts moving
toward the memory that follows $k$ in $K$.
(b) As $\alpha$ leaves $b_k$, the weights $a_{ij}(\alpha)$ change and the
phases move away from $\boldsymbol{\delta}^k$. The gate closes
($g\approx0$), so $\alpha$ is frozen until the phases have settled again.
(c) The phases lock at the next stored pattern and the gate reopens. Since
$r$ lags behind, it still indicates memory $k$ for a time of order
$\tau_r$, so $\alpha$ keeps pointing at the current pattern and the
network stays there. Once $r$ catches up, $\alpha$ moves toward the
following memory in $K$.
(d) Repeating (a)--(c) retrieves the stored patterns in the order
specified by $K$.}
\end{figure}

\section{Stability analysis of the three-timescale system}\label{sec:analysis}

The phases are the fastest variable, $\tau_\theta\ll\tau_r$, and the replicator, once the gate opens, moves on the timescale $\tau_\alpha/\beta$, which we take to be short compared with the recognition lag, $\tau_\alpha/\beta\ll\tau_r$. We therefore analyze the variables one at a time, holding the slower ones fixed. The fast layer was already treated in Section~\ref{sec:setting}. With $\alpha$ held at a vertex $b_k$, the phase equation is the Kuramoto network \eqref{eq: system dynamics} with weights $\boldsymbol a^k$, and by Theorem~\ref{thm: general graph} and Lemma~\ref{lemma: weights mapping} the phases converge to $\boldsymbol\delta^k$, the only stable equilibrium in $\mathcal R^k$. We assume throughout that this holds for every transition allowed by $K$, that is, $\boldsymbol\delta^k\in\mathcal R^j$ whenever $K_{kj}>0$. Once the phases have locked, $\dot{\boldsymbol\theta}=0$ and the gate is open, $g=1$. It therefore remains to analyze the recognition vector $r$ and the attention $\alpha$.

\subsection{Intermediate layer: the recognition lag}\label{sec:r}

Suppose the phases move from a pattern $\boldsymbol\delta^{i}$ to a new pattern $\boldsymbol\delta^k$ and lock there at time $t_0$. Because the phases are fast, this switch is almost instantaneous, and $\widetilde q$ switches with it. We first compute $\widetilde q$ after the switch. Setting $\boldsymbol\delta=\boldsymbol\delta^k$ in the definition of $\widetilde q$, the exponential for $j=k$ equals $1$, so $\widetilde q_k=1/\bigl(1+\sum_{j\neq k}e^{-\gamma\|\boldsymbol\delta^j-\boldsymbol\delta^k\|^2}\bigr)$ and
\[
1-\widetilde q_k=\frac{\sum_{j\neq k}e^{-\gamma\|\boldsymbol\delta^j-\boldsymbol\delta^k\|^2}}{1+\sum_{j\neq k}e^{-\gamma\|\boldsymbol\delta^j-\boldsymbol\delta^k\|^2}}\le\sum_{j\neq k}e^{-\gamma\|\boldsymbol\delta^j-\boldsymbol\delta^k\|^2}\le(M-1)\,e^{-\gamma d_{\min}^2},
\qquad
d_{\min}=\min_{j\neq k}\|\boldsymbol\delta^j-\boldsymbol\delta^k\| .
\]
Fix a small tolerance $\varepsilon\in(0,1)$, and take $\gamma$ large enough that $1-\widetilde q_k\le\varepsilon$. Then $\widetilde q\approx b_k$ for all $t\ge t_0$.

We next ask how long $r$ takes to reach this value. The recognition equation of \eqref{eq:joint} is linear in $r$, with solution
\begin{equation}\label{eq:r-solution}
r(t)=e^{-(t-t_0)/\tau_r}\,r(t_0)+\frac{1}{\tau_r}\int_{t_0}^{t}e^{-(t-s)/\tau_r}\,\widetilde q\bigl(\boldsymbol\theta(s)\bigr)\,ds .
\end{equation}
Before the switch $r$ had reached the old pattern, $r(t_0)\approx b_{i}$. Substituting this and $\widetilde q\approx b_k$ into \eqref{eq:r-solution} gives
\[
r(t)\approx e^{-(t-t_0)/\tau_r}\,b_{i}+\bigl(1-e^{-(t-t_0)/\tau_r}\bigr)\,b_k ,
\]
so $r_{i}$ decays and $r_k$ grows, both at rate $1/\tau_r$. Hence $r$ reaches the value of $\widetilde q$ up to the tolerance, $r_k\approx1-\varepsilon$, at about $t_0+\tau_r\log(1/\varepsilon)$, and until then it still points partly at the old pattern.

\subsection{Slow layer: routing by the attention}\label{sec:slow}

The slow dynamics of \eqref{eq:joint} has two parts, a replicator equation that uses $K$ and $r$ to choose the next memory, and the gate $g$ that decides when $\alpha$ may move. For the first part, take a memory $k$ with $\alpha_k>0$, so that $\log\alpha_k(t)$ is defined. Since $\tfrac{d}{dx}\log x=1/x$, differentiating the composition $\log\alpha_k(t)$ in time gives
\[
\frac{d}{dt}\log\alpha_k(t)=\frac{1}{\alpha_k(t)}\,\frac{d\alpha_k}{dt}(t)=\frac{\dot\alpha_k(t)}{\alpha_k(t)} .
\]
Dividing both sides of the attention equation by $\alpha_k$ and replacing $\dot\alpha_k/\alpha_k$ by $\tfrac{d}{dt}\log\alpha_k$ then gives
\[
\tau_\alpha\frac{d}{dt}\log\alpha_k=g\beta\bigl((K^\top r)_k-\langle K^\top r,\alpha\rangle\bigr).
\]
The average $\langle K^\top r,\alpha\rangle$ depends on the whole vector $\alpha$ but is the same for every $k$. Writing this equation for two memories $i$ and $j$ and subtracting therefore removes it, and since $\log\alpha_i-\log\alpha_j=\log(\alpha_i/\alpha_j)$ we obtain
\begin{equation}\label{eq:alpha-solution}
\frac{d}{dt}\log\frac{\alpha_i}{\alpha_j}=\frac{g\beta}{\tau_\alpha}\bigl((K^\top r)_i-(K^\top r)_j\bigr).
\end{equation}
We now use \eqref{eq:alpha-solution} to see where the attention goes. Since $(K^\top r)_j=\sum_{h}K_{hj}r_h$, if $r=b_h$ then $(K^\top r)_j=K_{hj}$, the entry of $K$ in row $h$ and column $j$. Suppose the attention has just moved to memory $k$ from a memory $i$, and the phases have locked at $\boldsymbol\delta^k$, so that $g=1$. In $K$, memory $k$ follows memory $i$, so $K_{ik}$ is the largest entry in row $i$. Let $j$ be the memory that follows $k$, so that $K_{kj}$ is the largest entry in row $k$. We apply \eqref{eq:alpha-solution} to the pair $j,k$.

Right after the phases lock, $r$ still points at the previous memory, $r\approx b_{i}$, by Section~\ref{sec:r}. Then
\[
\frac{d}{dt}\log\frac{\alpha_{j}}{\alpha_k}\approx\frac{\beta}{\tau_\alpha}\bigl(K_{ij}-K_{ik}\bigr)<0,
\]
because $K_{ik}$ is the largest entry in row $i$. The ratio $\alpha_{j}/\alpha_k$ decreases, and the attention stays at memory $k$.

Once $r$ has reached the new pattern, $r\approx b_k$, and
\[
\frac{d}{dt}\log\frac{\alpha_{j}}{\alpha_k}\approx\frac{\beta}{\tau_\alpha}\bigl(K_{kj}-K_{kk}\bigr)>0 .
\]
The ratio $\alpha_{j}/\alpha_k$ now grows exponentially, and the attention moves from memory $k$ to memory $j$. Thus the attention stays at the current memory while $r$ lags and moves to the next memory once $r$ has caught up. Note that this argument needs $\alpha_{j}>0$. Because $\dot\alpha_{j}$ is proportional to $\alpha_{j}$, if $\alpha_{j}=0$ it stays $0$, so a small amount of attention on the next memory must be inherited from the previous transition or from noise.

For the second part, the gate multiplies the whole right-hand side of the attention equation. While the phases are still moving toward a pattern, $g\approx0$ and $\alpha$ is frozen, so the attention cannot move before the fast dynamics has reached a memory pattern. Once the phases lock, $g=1$ and the replicator acts as described above.

\subsection{Combined behavior: dwell time}\label{sec:combined}

What we care about is how long the phases stay at each memory pattern, the dwell time. The three layers combine as follows. The phases are fast and sit at whichever pattern the attention selects, so they stay at $\boldsymbol\delta^k$ as long as $\alpha$ stays at memory $k$. By Section~\ref{sec:slow}, $\alpha$ stays at memory $k$ while $r$ still points at the previous memory, and by Section~\ref{sec:r}, $r$ needs a time of about $\tau_r\log(1/\varepsilon)$ to catch up with $\boldsymbol\delta^k$. The dwell time of the phases is therefore set mainly by $r$. After $r$ has caught up, the replicator needs a short additional time to move $\alpha$ to the next memory $j$, and only then do the phases leave $\boldsymbol\delta^k$ (see Figure~\ref{fig:orbit}(c)). We measure the dwell time at memory $k$ from the time $t_0$ at which the phases lock at $\boldsymbol\delta^k$ until $\alpha_j$ overtakes $\alpha_k$. The next proposition bounds it.

\begin{proposition}[Dwell time]\label{prop:dwell}
Let $K$ have entries in $[0,1]$. Fix a memory $k$ and let $j$ be the memory that follows it, so that $K_{kj}$ is the largest off-diagonal entry in row $k$ and $K_{kj}>K_{kk}$. Let $0<\varepsilon<(K_{kj}-K_{kk})/4$, and suppose the phases lock at $\boldsymbol\delta^k$ at time $t_0$. Assume that, as long as $\alpha_k>\alpha_j$, the gate is open ($g=1$) and $1-\widetilde q_k(\boldsymbol\theta)\le\varepsilon$, and that $\alpha_j\ge\varepsilon\,\alpha_k$ at time $t_0+\tau_r\log(1/\varepsilon)$. Then $\alpha_j$ overtakes $\alpha_k$ no later than
\[
t_0+\tau_r\log\frac1\varepsilon+\frac{\tau_\alpha\log(1/\varepsilon)}{\beta\,(K_{kj}-K_{kk}-4\varepsilon)} .
\]
To leading order in $\varepsilon$, the dwell time at memory $k$ is therefore
\begin{equation}\label{eq:dwell}
\tau_r\log\frac{1}{\varepsilon}+\frac{\tau_\alpha\log(1/\varepsilon)}{\beta\,(K_{kj}-K_{kk})},
\end{equation}
whose denominator reduces to $\beta$ for a permutation $K$.
\end{proposition}

\begin{proof}
Let $t_1=t_0+\tau_r\log(1/\varepsilon)$ and suppose $\alpha_k>\alpha_j$ on $[t_0,t]$ with $t>t_1$. Since $1-\widetilde q_k\le\varepsilon$, the recognition equation gives $\tau_r\tfrac{d}{dt}(1-r_k)\le\varepsilon-(1-r_k)$, so $1-r_k\le\varepsilon+e^{-(t-t_0)/\tau_r}$ and hence $r_k\ge1-2\varepsilon$ on $[t_1,t]$, and since $|K_{ij}-K_{ik}|\le1$ and $K_{kj}-K_{kk}\le1$,
\[
(K^\top r)_j-(K^\top r)_k=\sum_{i}r_i\bigl(K_{ij}-K_{ik}\bigr)\ge r_k\bigl(K_{kj}-K_{kk}\bigr)-(1-r_k)\ge K_{kj}-K_{kk}-4\varepsilon .
\]
With $g=1$, \eqref{eq:alpha-solution} shows that $\log(\alpha_j/\alpha_k)$ grows at rate at least $\beta(K_{kj}-K_{kk}-4\varepsilon)/\tau_\alpha$. It starts above $\log\varepsilon$ at $t_1$, so it reaches $0$ within the stated time.
\end{proof}

The first term of \eqref{eq:dwell} is the time $r$ needs to catch up with the new pattern, during which $\alpha$, and with it the phases, is held at memory $k$. The second term is the additional time the replicator needs to move $\alpha$ to the next memory once $r$ has caught up. For a permutation $K$ it equals $\tau_\alpha\log(1/\varepsilon)/\beta$. For a general $K$ the push toward the next memory is weaker by the factor $K_{kj}-K_{kk}$, and this term is longer. The ratio of the second term to the first is $\tau_\alpha/\bigl(\beta\tau_r(K_{kj}-K_{kk})\bigr)$, so when $\tau_\alpha/\beta\ll\tau_r$ the dwell time of the phases is set mainly by $r$ and is approximately $\tau_r\log(1/\varepsilon)$. If the gate is only partly open, $g\ge g_0>0$, the second term is divided by $g_0$.

\section{Sequential memory capacity}\label{sec:capacity}

We now ask how many patterns the system can store and retrieve in sequence. We use a probabilistic notion of capacity, in which the patterns are drawn at random and the capacity is the largest number $M$ of patterns for which each pattern is retrieved within a tolerance $\varepsilon$ with probability at least $1-\eta$.

We first identify what limits $M$. The fast layer stores every pattern exactly. By Lemma~\ref{lemma: weights mapping}, each $\boldsymbol\delta^k$ is a stable equilibrium under its own weights $\boldsymbol a^k$, and the attention selects one set of weights at a time, so the stored patterns do not interfere with one another as they do in a Hopfield network, where all patterns share one weight matrix. The routing is set by $K$ and does not depend on $M$. What can fail as $M$ grows is recognition. When the patterns crowd together, if the softmax $\widetilde q$ at $\boldsymbol\delta^k$ gives non-negligible weight to the neighboring patterns, then $r$ no longer reaches $b_k$, and the attention cannot tell which memory comes next. We therefore say that pattern $k$ is recognized if $1-\widetilde q_k(\boldsymbol\delta^k)\le\varepsilon$, which is the condition used in Section~\ref{sec:analysis}, and we define the capacity as
\[
C(\varepsilon,\eta)=\max\Bigl\{M:\ \Pr\bigl[\,1-\widetilde q_k(\boldsymbol\delta^k)\le\varepsilon\,\bigr]\ge1-\eta\ \text{ for every }k\Bigr\}.
\]
In words, the capacity is the largest number of patterns for which each pattern is recognized with probability at least $1-\eta$, where the probability is over the random draw of the patterns described below. 

According to Section~\ref{sec:setting}, the memory patterns lie in the reachable set $\mathcal R$ of Theorem~\ref{thm: general graph} and satisfy conditions~(ii) and~(iii) of Lemma~\ref{lemma: weights mapping}. On the honeycomb graph this set splits over the cycles. All patterns in $\mathcal R$ have the same sum of phase differences around each cycle $\mathcal C_\ell$, namely $2\pi w_\ell$ for an integer winding number $w_\ell$. We fix $\boldsymbol w=(w_1,\dots,w_m)$ with $0<|w_\ell|<n_c/4$. For one cycle with winding number $w$, let $P_w\subset\mathbb R^{n_c}$ be the set of phase differences $(\delta_1,\dots,\delta_{n_c})$ on the edges of the cycle such that every $\delta_i$ is nonzero and lies in $(-\tfrac\pi2,\tfrac\pi2)$, all $\delta_i$ have the same sign, and $\delta_1+\dots+\delta_{n_c}=2\pi w$. To analyze the capacity, we first check how many memory patterns a single cycle can hold while remaining distinguishable by $\widetilde q$.

The set $P_w$ contains all the patterns that one cycle can store. Because the $n_c$ phase differences must add up to $2\pi w$, only $n_c-1$ of them are free, so $P_w$ has dimension $n_c-1$. A point of $P_w$ is determined by its first $n_c-1$ phase differences, since the last one is $\delta_{n_c}=2\pi w-\delta_1-\dots-\delta_{n_c-1}$. We define the volume of $P_w$ by integrating over these $n_c-1$ free phase differences, which gives
\begin{equation}\label{eq:Vw-def}
V_w=\sqrt{n_c}\int_{\mathbb R^{n_c-1}}\mathbf 1_{P_w}\bigl(\delta_1,\dots,\delta_{n_c-1},\,2\pi w-\delta_1-\dots-\delta_{n_c-1}\bigr)\,d\delta_1\cdots d\delta_{n_c-1},
\end{equation}
where $\mathbf 1_{P_w}$ equals $1$ on $P_w$ and $0$ elsewhere. The integral counts, with the ordinary volume of $\mathbb R^{n_c-1}$, the choices of the free phase differences that give a point of $P_w$. The factor $\sqrt{n_c}$ is there because $P_w$ lies in the hyperplane $\delta_1+\dots+\delta_{n_c}=2\pi w$, which is tilted with respect to the coordinate axes, and with this factor $V_w$ is the $(n_c-1)$-dimensional volume of $P_w$ in the Euclidean distance of $\mathbb R^{n_c}$. We define the volume in this way for two reasons. First, it defines what it means to draw a pattern uniformly at random, since the probability that the pattern falls in a region of $P_w$ is the volume of that region divided by $V_w$. Second, it is based on the same distance $\|\boldsymbol\delta^j-\boldsymbol\delta^k\|$ that the softmax $\widetilde q$ uses to compare patterns, so the size of $P_w$ and the size of the region around a pattern in which the softmax confuses it with another are measured in the same units. A larger $V_w$ means more room for patterns.

Now that we know the volume available on each cycle of the honeycomb graph, we analyze how many patterns can fit into this volume at the resolution of $\widetilde q$. When the phases sit at $\boldsymbol\delta^k$, the softmax in $\widetilde q$ gives another pattern at distance $d$ the weight $e^{-\gamma d^2}$ relative to the weight $1$ of pattern $k$. More precisely, by Section~\ref{sec:r} we have $1-\widetilde q_k(\boldsymbol\delta^k)=S_k/(1+S_k)$ with $S_k=\sum_{j\neq k}e^{-\gamma\|\boldsymbol\delta^j-\boldsymbol\delta^k\|^2}$, so pattern $k$ is recognized, that is, $1-\widetilde q_k(\boldsymbol\delta^k)\le\varepsilon$, if and only if $S_k\le\varepsilon/(1-\varepsilon)$. To see what this condition means, suppose there is only one other pattern, at distance $d$ from $\boldsymbol\delta^k$. Then $S_k=e^{-\gamma d^2}$, and pattern $k$ is recognized if and only if $e^{-\gamma d^2}\le\varepsilon/(1-\varepsilon)$. Taking logarithms, this is the same as
\[
d\ge d_\varepsilon=\frac{1}{\sqrt\gamma}\sqrt{\log\frac{1-\varepsilon}{\varepsilon}},
\]
where we assume $\varepsilon<\tfrac12$ so that the logarithm is positive. Thus $d_\varepsilon$ is the smallest distance at which the softmax can separate two patterns. If the other pattern is closer than $d_\varepsilon$, then the softmax gives it more than the share $\varepsilon$ and pattern $k$ is not recognized. If it is at distance $d_\varepsilon$ or more, then pattern $k$ is recognized. The distance $d_\varepsilon$ is a fixed multiple of $1/\sqrt\gamma$, so a sharper softmax separates closer patterns. For example, $\gamma=50$ and $\varepsilon=0.05$ give $d_\varepsilon\approx0.24$ radians. Now consider one cycle and let the other pattern $\boldsymbol\delta$ be drawn uniformly from $P_w$. Its expected weight in the softmax is
\[
\mathbb E\,e^{-\gamma\|\boldsymbol\delta-\boldsymbol\delta^k\|^2}=\frac{1}{V_w}\int_{P_w}e^{-\gamma\|\boldsymbol\delta-\boldsymbol\delta^k\|^2}\,d\boldsymbol\delta\le\frac{1}{V_w}\Bigl(\frac{\pi}{\gamma}\Bigr)^{(n_c-1)/2},
\]
where the integral is taken with the volume of \eqref{eq:Vw-def}, and the bound holds because the same Gaussian integrated over the whole hyperplane equals $(\pi/\gamma)^{(n_c-1)/2}$. In this sense each pattern occupies a volume $(\pi/\gamma)^{(n_c-1)/2}$ of $P_w$, and the expected total weight $\mathbb E\,S_k$ of $M-1$ other patterns is at most $(M-1)\big/\bigl(V_w(\gamma/\pi)^{(n_c-1)/2}\bigr)$. It stays small as long as $M$ is small compared with $V_w(\gamma/\pi)^{(n_c-1)/2}$.

The number of patterns that one cycle can hold is therefore governed by the quantity $V_w(\gamma/\pi)^{(n_c-1)/2}$, and to evaluate it we need the volume $V_w$. The next lemma gives $V_w$ in closed form. Its proof is in the Appendix.

\begin{lemma}[Volume of the pattern set on one cycle]\label{lem:volume}
For every integer $w$ with $0<|w|<n_c/4$,
\begin{equation}\label{eq:volume}
V_w=\sqrt{n_c}\,\Bigl(\frac\pi2\Bigr)^{n_c-1}\frac{1}{(n_c-1)!}\sum_{s=0}^{4|w|}(-1)^s\binom{n_c}{s}\bigl(4|w|-s\bigr)^{n_c-1}.
\end{equation}
\end{lemma}

We now pass from one cycle to the whole network. The honeycomb graph has $m$ cycles, and a pattern assigns one point of $P_{w_\ell}$ to each cycle $\mathcal C_\ell$. The pattern set is therefore the product $P_{\boldsymbol w}=P_{w_1}\times\dots\times P_{w_m}$, of dimension $m(n_c-1)$, and we draw the patterns $\boldsymbol\delta^1,\dots,\boldsymbol\delta^M$ independently and uniformly from it. The squared distance between two patterns is the sum of their squared distances on the $m$ cycles, so the weight $e^{-\gamma\|\boldsymbol\delta^j-\boldsymbol\delta^k\|^2}$ is a product of one factor per cycle, and the quantities $V_{w_\ell}(\gamma/\pi)^{(n_c-1)/2}$ of the individual cycles multiply. This leads to the following lower bound on the capacity.

\begin{theorem}[Sequential capacity]\label{thm:capacity}
For every $\gamma>0$ and $\varepsilon,\eta\in(0,1)$,
\begin{equation}\label{eq:cap-seq}
C(\varepsilon,\eta)\;\ge\;1+\varepsilon\,\eta\,\prod_{\ell=1}^{m}V_{w_\ell}\bigl(\tfrac{\gamma}{\pi}\bigr)^{(n_c-1)/2}.
\end{equation}
\end{theorem}

\begin{proof}
Fix $k$. By Section~\ref{sec:r}, $1-\widetilde q_k(\boldsymbol\delta^k)\le\sum_{j\neq k}e^{-\gamma\|\boldsymbol\delta^j-\boldsymbol\delta^k\|^2}$. Markov's inequality therefore gives
\[
\Pr\bigl[\,1-\widetilde q_k(\boldsymbol\delta^k)>\varepsilon\,\bigr]\le\frac1\varepsilon\sum_{j\neq k}\mathbb E\,e^{-\gamma\|\boldsymbol\delta^j-\boldsymbol\delta^k\|^2}=\frac{M-1}{\varepsilon}\,\mathbb E\,e^{-\gamma\|X-Y\|^2},
\]
where $X$ and $Y$ are independent and uniform on $P_{\boldsymbol w}$. The  squared distance $\|X-Y\|^2$ is a sum over the $m$ cycles, and the cycles are independent, so the expectation is the product of $m$ factors, one per cycle. For a cycle with winding number $w$, let $x,y$ be independent and uniform on $P_w$. Every difference $x-y$ lies in the $(n_c-1)$-dimensional subspace $H=\{u\in\mathbb R^{n_c}:u_1+\dots+u_{n_c}=0\}$, so
\[
\mathbb E\,e^{-\gamma\|x-y\|^2}=\frac{1}{V_w^2}\int_{P_w}\int_{P_w} e^{-\gamma\|x-y\|^2}\,dy\,dx\le\frac{1}{V_w^2}\int_{P_w}\int_H e^{-\gamma\|u\|^2}\,du\,dx=\frac1{V_w}\Bigl(\frac\pi\gamma\Bigr)^{(n_c-1)/2},
\]
where the last step is the Gaussian integral over $H$. Hence the failure probability is at most $\frac{M-1}{\varepsilon}\prod_{\ell=1}^{m}V_{w_\ell}^{-1}(\pi/\gamma)^{(n_c-1)/2}$, which is at most $\eta$ whenever $M$ satisfies \eqref{eq:cap-seq}.
\end{proof}

\begin{corollary}[Exponential growth in the number of cycles]\label{cor:exp}
Suppose that all cycles have the same winding number $w$ and that $\gamma>\pi\,V_w^{-2/(n_c-1)}$. Then
\begin{equation}\label{eq:cap-exp}
C(\varepsilon,\eta)\;\ge\;1+\varepsilon\,\eta\,e^{\kappa m},
\qquad
\kappa=\log\Bigl(V_w\bigl(\tfrac{\gamma}{\pi}\bigr)^{(n_c-1)/2}\Bigr)>0 .
\end{equation}
That is, the capacity grows at least exponentially in the number of cycles $m$, with rate $\kappa$ per cycle.
\end{corollary}

\begin{proof}
When all winding numbers are equal to $w$, the product in \eqref{eq:cap-seq} is the $m$-th power of $V_w(\gamma/\pi)^{(n_c-1)/2}=e^{\kappa}$. The condition on $\gamma$ is equivalent to $V_w(\gamma/\pi)^{(n_c-1)/2}>1$, that is, to $\kappa>0$.
\end{proof}

In the bound \eqref{eq:cap-seq}, cycle $\mathcal C_\ell$ contributes the factor $V_{w_\ell}(\gamma/\pi)^{(n_c-1)/2}$. If the cycles have different winding numbers, \eqref{eq:cap-exp} still holds with $\kappa$ replaced by the smallest of the rates of the individual cycles. Since $m=N/n_c$, for a fixed cycle length $n_c$ the bound is also exponential in the number of edges $N$. The factor of a cycle with winding number $w$ exceeds $1$ once the sharpness satisfies $\gamma>\pi V_w^{-2/(n_c-1)}$, and it grows with $\gamma$, since a sharper softmax tells closer patterns apart. For example, for $n_c=5$ the only admissible winding numbers are $w=\pm1$, with $V_{\pm1}\approx0.57$, and the sharpness $\gamma=50$ used in Section~\ref{sec:experiments} gives a factor of about $1.4\times10^2$ per cycle, that is, $\kappa\approx5.0$.

The volume $V_w$ depends on the cycle length through the winding number. For a fixed $w$ it is not monotone in $n_c$. The phase differences must add up to $2\pi w$, so on a long cycle most of them are forced to be small, and $V_w$ eventually decreases to zero. For $w=1$, for example, $V_1$ increases from $0.57$ at $n_c=5$ to about $45$ at $n_c=10$ and then decreases, falling below $0.05$ at $n_c=20$. The volume is largest when the phase differences can spread over the whole interval, which happens when their average $2\pi|w|/n_c$ is near $\pi/4$, that is, $|w|\approx n_c/8$. For this choice the sum of $n_c$ independent uniform variables on $(0,1)$ that appears in the proof of Lemma~\ref{lem:volume} is evaluated at its mean $n_c/2$. Since this sum has variance $n_c/12$, the central limit theorem gives its density at the mean as approximately $\sqrt{6/(\pi n_c)}$, and therefore $V_w\approx\sqrt{6/\pi}\,(\pi/2)^{n_c-1}$, which grows exponentially in $n_c$. Longer cycles therefore give more room per cycle provided the winding number grows in proportion to the cycle length.

\section{Numerical experiments}\label{sec:experiments}

In this section we use numerical simulations to test the theoretical analysis of the previous sections. We first check that the three-timescale dynamics \eqref{eq:joint} make the Kuramoto network traverse, on its own, the memory sequence set by the routing matrix $K$. We then test the dwell-time analysis, that is, whether the time the dynamics stay at each memory pattern is correctly predicted by Proposition~\ref{prop:dwell}. Finally, we introduce context-dependent routing, in which the routing matrix $K$ depends on an external input, so that the network can change the order of recall in real time instead of following one prescribed order.

\subsection{Cyclic routing}\label{sec:exp-honeycomb}

We simulate \eqref{eq:joint} on the honeycomb graph $\mathcal{G}_5^5$ of Definition~\ref{def:honeycomb}, which consists of $m=5$ cycles with $n_c=5$ edges each, where consecutive cycles share one node, so that there are $n=m(n_c-1)+1=21$ nodes and $N=mn_c=25$ edges. We store $M=10$ patterns, labeled $0$ to $9$, with the weights given by \eqref{eq: weights mapping}. The routing matrix $K$ is the cyclic permutation $0\to1\to\dots\to9\to0$. The system starts near pattern $0$ with small Gaussian phase noise of standard deviation $0.1$ rad. The parameters are $\gamma=50$, $\tau_r=100$, $\tau_\alpha=80$, $\beta=15$, $\delta_g=10^{-2}$, and $\tau_\theta=0.1$. Figure~\ref{fig:honeycomb_sim} shows one full traversal after the initial transient.

Panel~(a) shows the active pattern $j^*(t)=\arg\max_j\alpha_j(t)$. All $10$ patterns are visited in the prescribed order, and each is held for about the same time. Panel~(b) shows the $25$ phase differences $\delta_e(t)$ together with their targets $\delta_e^{j^*(t)}$. While a pattern is active, the phase differences coincide with the targets, so the fast layer retrieves each stored pattern. Panel~(c) shows the residual $\|\boldsymbol\delta-\boldsymbol\delta^{j^*}\|_2$ on a logarithmic scale. After each transition it decays to about $10^{-5}$ or less, and it exceeds $0.05$ only during the short transitions between patterns. This agrees with Theorem~\ref{thm: general graph}, which guarantees that the phases converge to the next pattern after each switch of the weights.

\begin{figure}[H]
    \centering
    \includegraphics[width=\linewidth]{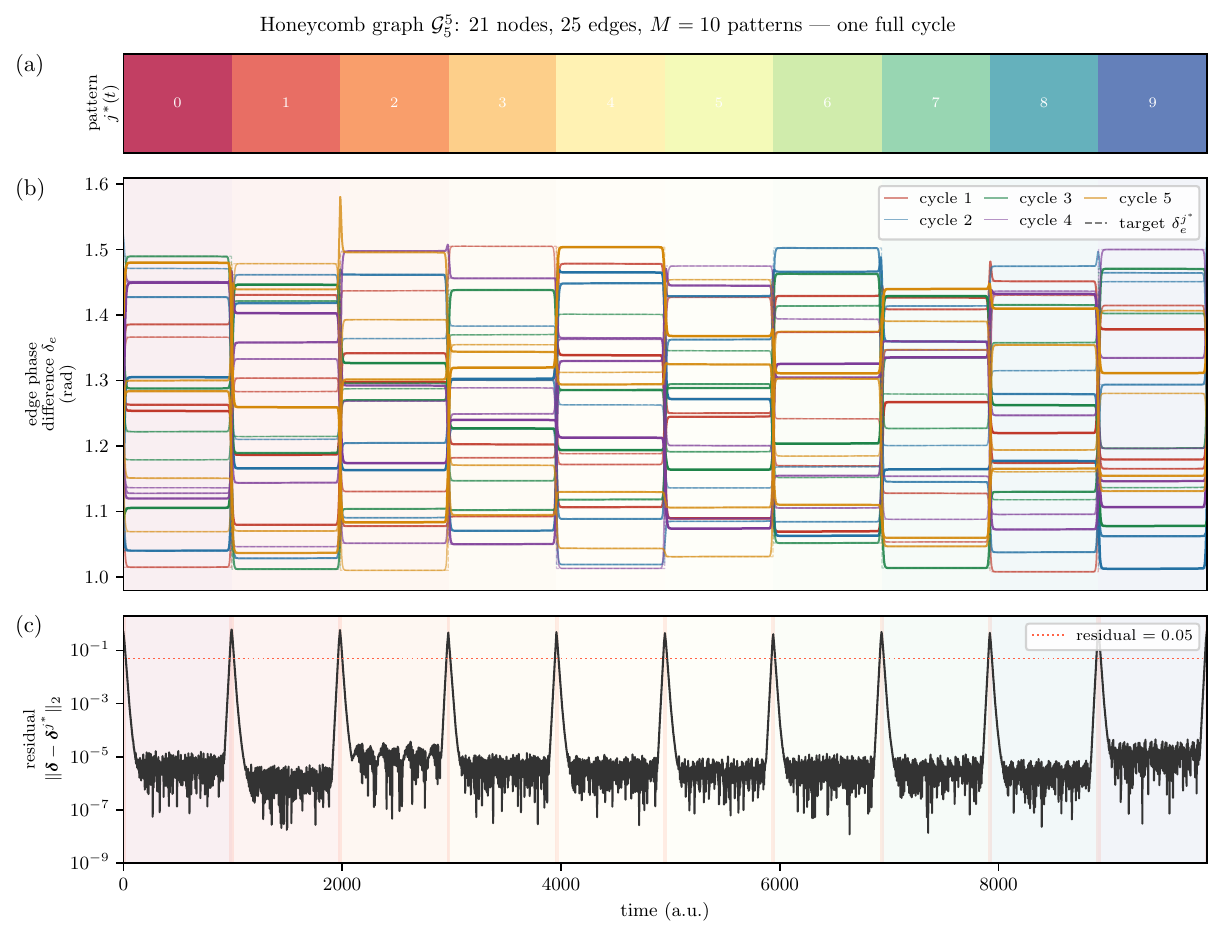}
    \caption{Simulation of the three-timescale system \eqref{eq:joint} on $\mathcal{G}_5^5$ with $M=10$ patterns and cyclic routing. One full traversal is shown, after the initial transient. Panel~(a) shows the active pattern $j^*(t)$, with one color per stored pattern. Panel~(b) shows the $25$ phase differences $\delta_e(t)$ (solid, colored by cycle) and their targets $\delta_e^{j^*}$ (dashed). The phase differences lock onto the targets while a pattern is active. Panel~(c) shows the residual $\|\boldsymbol\delta-\boldsymbol\delta^{j^*}\|_2$ on a logarithmic scale. The dotted line marks the value $0.05$, which the residual exceeds only during transitions.}
    \label{fig:honeycomb_sim}
\end{figure}

\subsection{Dwell time}\label{sec:exp-dwell}

We next test Proposition~\ref{prop:dwell}. We vary $\tau_r$ and keep $\beta=15$ and $\tau_\alpha=80$ fixed, on a single cycle of $6$ oscillators with $M=5$ patterns, a cyclic permutation $K$ with zero diagonal, and tolerance $\varepsilon=10^{-2}$. For each $\tau_r$ we compare the mean dwell time in the simulation with the prediction \eqref{eq:dwell}. For a permutation $K$ the prediction is $\tau_r\log(1/\varepsilon)+\tau_\alpha\log(1/\varepsilon)/\beta$. Figure~\ref{fig:dwell_time} shows the result.

The measured dwell time grows linearly with $\tau_r$, as \eqref{eq:dwell} predicts. The prediction lies above the measured values for every $\tau_r$, by about $10\%$. This is expected, because Proposition~\ref{prop:dwell} is an upper bound. Its proof treats the two stages one after the other: first $r$ catches up with the new pattern to within $\varepsilon$, and only then does the replicator move the attention, starting from the smallest value of $\alpha_j$ that the proposition allows. In the simulation the two stages overlap, since the replicator starts to move the attention as soon as $r$ favors the next memory, before $r$ has fully caught up. The formula \eqref{eq:dwell} therefore gives the correct dependence on $\tau_r$ and a slightly conservative value, which is enough to choose $\tau_r$ for a desired dwell time.

\begin{figure}[H]
    \centering
    \includegraphics[width=0.75\linewidth]{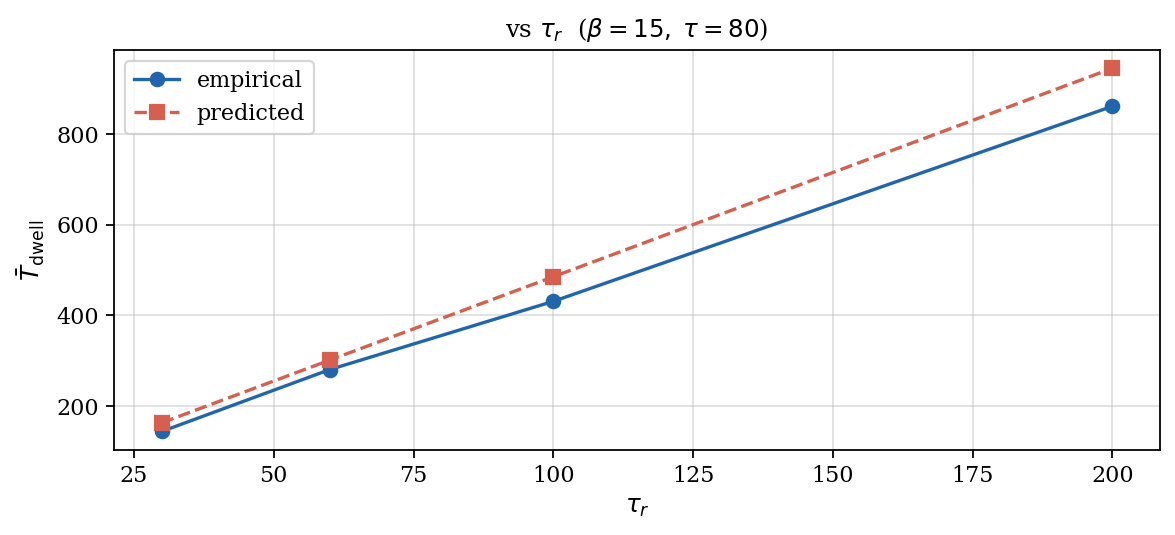}
    \caption{Mean dwell time as a function of $\tau_r$, with $\beta=15$ and $\tau_\alpha=80$ fixed (a single cycle of $6$ oscillators, $M=5$, cyclic routing). The dashed line with squares is the prediction \eqref{eq:dwell}, and the solid line with circles is the mean dwell time in the simulation. Both grow linearly with $\tau_r$. The prediction is an upper bound and lies about $10\%$ above the measured values.}
    \label{fig:dwell_time}
\end{figure}

\subsection{Context-dependent routing}\label{sec:exp-context}

In our system, storage and routing are carried by different layers. The patterns are stored in the fast layer through the weights $\boldsymbol a^k$, and the order of recall is set only by the routing matrix $K$ in the slow layer. The same stored patterns can therefore be recalled in a different order by changing $K$, without changing any oscillator weight. This is not possible in a network that encodes the order in its coupling weights, where a new order requires new weights.

To make the order depend on a context, we replace $K$ by a set of routing matrices $K^{(0)},\dots,K^{(S-1)}$, one for each context, and let an external context signal $u(t)$ in the simplex $\Delta^{S}$ select among them. The dynamics \eqref{eq:joint} are unchanged except that $K$ is replaced by
\begin{equation}
  K(u)\;=\;\sum_{s=0}^{S-1}u_{s}\,K^{(s)},
  \label{eq:context_T}
\end{equation}
so that the fitness in the replicator equation becomes $\beta\,K(u)^{\top}r$.

We illustrate this with $M=4$ stored patterns, labeled $0$ to $3$, and two contexts, which we call weekday and weekend. In the weekday context the patterns are visited in the order $0\to1\to2\to3\to0$, and in the weekend context in the order $0\to2\to3\to1\to0$. We use the same honeycomb graph $\mathcal{G}_5^5$ as in Section~\ref{sec:exp-honeycomb}, with $\gamma=50$, $\beta=15$, $\delta_g=10^{-2}$, $\tau_r=200$, $\tau_\alpha=80$, and $\tau_\theta=0.1$. In this experiment the two routing matrices, $K^{(0)}$ (weekday) and $K^{(1)}$ (weekend), are sampled at random, so that they are not exact permutations: for each pattern, the weights of its possible successors are drawn from a Dirichlet distribution with concentration $8$ for the intended successor and $1$ for each other pattern, and the diagonal is zero. In the resulting matrices the intended successors have weights between $0.59$ and $0.95$, and the other entries lie between $0.01$ and $0.40$ (Figure~\ref{fig:context_routing}, right column).

In the three panels of Figure~\ref{fig:context_routing} the oscillator weights, the stored patterns, and all time constants are the same, and only the routing matrix differs. In Panel~(a) the routing matrix is $K^{(0)}$, and the system visits the patterns in the weekday order, with a mean dwell time of about $460$ time units. In Panel~(b) the routing matrix is $K^{(1)}$, and the same network visits the patterns in the weekend order, with a mean dwell time of about $555$ time units. In both panels the attention $\alpha_j(t)$ of the active pattern is close to $1$ during each dwell, and the transitions between patterns are sharp. The dwell times differ from one transition to the next, because the routing entries differ from one pattern to the next. They are shorter than the upper bound of Proposition~\ref{prop:dwell}, which is about $950$ time units for $\varepsilon=10^{-2}$.

In Panel~(c) the routing matrix is switched from $K^{(0)}$ to $K^{(1)}$ at time $t=3200$. Only the fitness $\beta K^\top r$ changes at the switch, and the phases, the recognition $r$, and the attention $\alpha$ remain continuous. At the switch the system is at pattern $3$, whose successor is $0$ on a weekday and $1$ on a weekend. The system moves to pattern $1$ and then continues in the weekend order, $1\to0\to2\to3$. Over the whole run the dwell times have mean about $490$ and standard deviation about $175$ time units, and the system needs no warm-up period after the switch. The order of recall has therefore been changed during the run without changing any stored pattern or oscillator weight.

\begin{figure}[H]
  \centering
  \includegraphics[width=\linewidth]{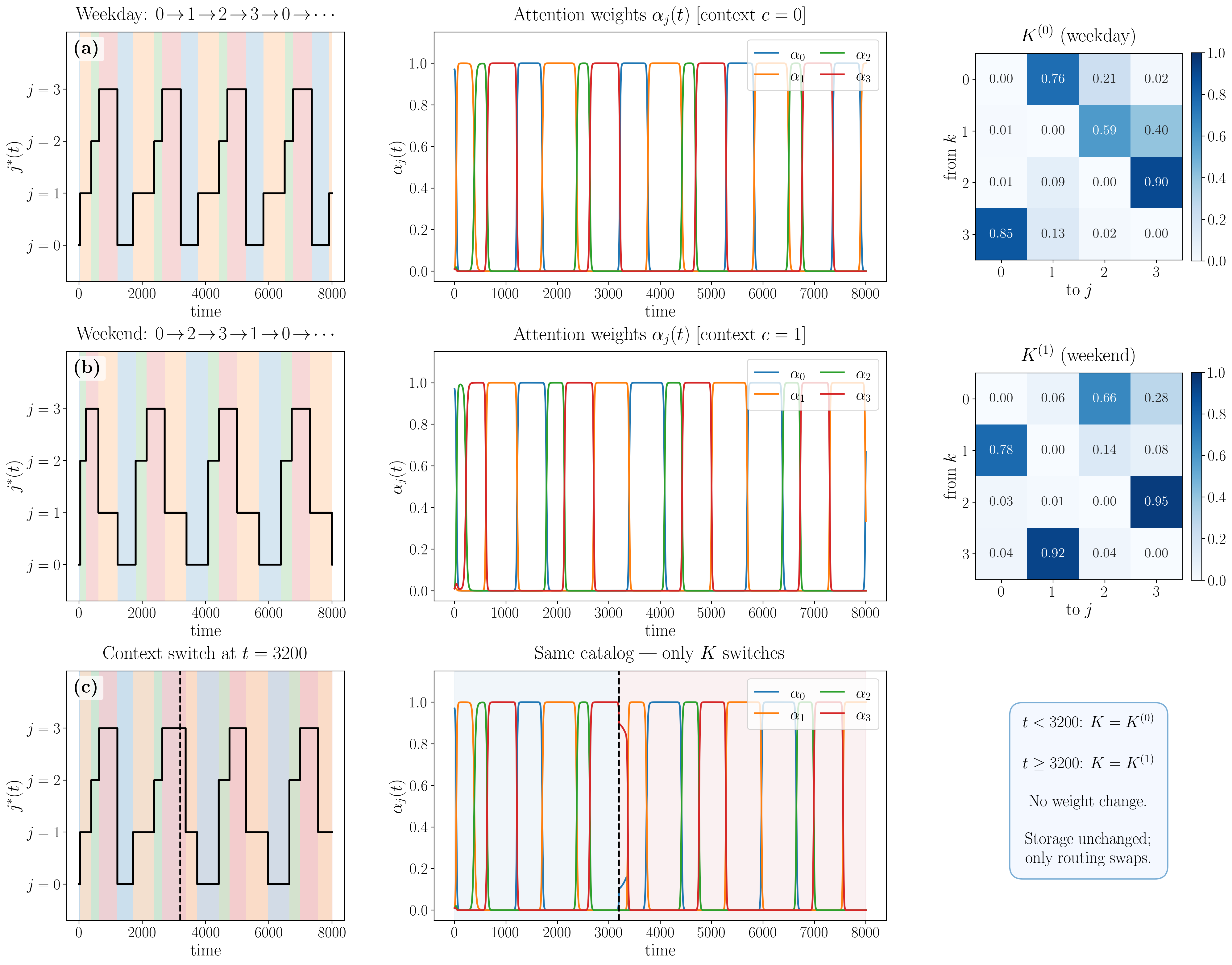}
  \caption{Context-dependent routing with $M=4$ patterns on $\mathcal{G}_5^5$. The parameters are those of Figure~\ref{fig:honeycomb_sim} except $\tau_r=200$. In each row, the left plot shows the active pattern $j^*(t)=\arg\max_j\alpha_j(t)$, with the background colored by the active pattern, the center plot shows the attention $\alpha_j(t)$, and the right plot shows the routing matrix, which has zero diagonal. In row~(a) the routing matrix is $K^{(0)}$ (weekday), and the patterns are visited in the order $0\to1\to2\to3\to0$ with a mean dwell time of about $460$. In row~(b) the routing matrix is $K^{(1)}$ (weekend), and the same patterns, with the same oscillator weights, are visited in the order $0\to2\to3\to1\to0$ with a mean dwell time of about $555$. In row~(c) the routing matrix is switched from $K^{(0)}$ to $K^{(1)}$ at $t=3200$ (dashed line). Before the switch the system follows the weekday order. After the switch it moves from pattern $3$ to pattern $1$ instead of pattern $0$ and continues in the weekend order.}
  \label{fig:context_routing}
\end{figure}

The experiment uses two pure contexts, $u=b_0$ and $u=b_1$. The construction \eqref{eq:context_T} also allows a mixed signal $u=(u_0,u_1)$, which gives the routing matrix $u_0K^{(0)}+u_1K^{(1)}$, and more generally any routing matrix $K$. Different choices of $K$ give different kinds of traversal, such as a fixed cycle, a choice among several successors, or separate sub-sequences that do not mix. In all of these cases the stored patterns and the oscillator weights are unchanged. We leave the study of these cases to future work.

\section{Discussion and future work}\label{sec:discussion}

In this paper, we present a sequential memory built from a network of Kuramoto oscillators. It stores a set of memories and recalls them continuously, one after another, in a programmable order that can be learned, without external control. Each memory pattern is a stable phase-locked configuration of the honeycomb graph, and the weights that store it are given in closed form by Lemma~\ref{lemma: weights mapping}. Specifically, we introduce dynamics with three separated timescales that produce the memory sequence. The fast layer holds the network at the current memory. The intermediate layer recognizes, with a delay, which memory the network is at (Section~\ref{sec:r}). The slow layer uses this delayed recognition to move the attention to the next memory prescribed by the routing matrix $K$ (Section~\ref{sec:slow}). Together, the three layers give the bound \eqref{eq:dwell} on the time spent at each memory, which the simulations of Section~\ref{sec:experiments} confirm as a slightly conservative upper bound.

The main feature of this design is that storage and routing are separate. The memories are stored in the oscillator weights, and the order of recall is stored only in the programmable routing matrix $K$. The stored memories therefore do not interfere with one another, and their number is limited only by how well the recognition layer can tell them apart. This gives the capacity bound of Theorem~\ref{thm:capacity}, which grows exponentially in the number of cycles. The separation also means that the order of recall can be changed, even during a run, by changing $K$ and nothing else, as shown in Section~\ref{sec:exp-context}. In networks that encode the order in the coupling weights, by contrast, a new order requires new weights.

Several directions remain open. The first is learning. Here the memory patterns and the routing matrix are given, and a natural next step is to learn them from data while keeping the guarantees of the analysis. For the routing matrix, one candidate is a Hebbian rule that strengthens the entry linking two memories that are recalled in succession. The second direction is context. In Section~\ref{sec:exp-context} the context signal is supplied from outside and selects one routing matrix at a time. It remains to analyze mixed contexts, in which \eqref{eq:context_T} combines several routing matrices, and to let the system infer the context from its own recognition $r(t)$, so that no external signal is needed. The third direction is capacity. Theorem~\ref{thm:capacity} is a lower bound, and it remains to complement it with an upper bound and to find the cycle length and winding number that give the largest capacity. The last is to extend the closed-form weight design from the honeycomb graph to more general graphs, for which Theorem~\ref{thm: general graph} already gives the conditions for a transition.
\bibliographystyle{unsrtnat}
\bibliography{references}

\begin{thebibliography}{20}
\providecommand{\natexlab}[1]{#1}
\providecommand{\url}[1]{\texttt{#1}}
\expandafter\ifx\csname urlstyle\endcsname\relax
  \providecommand{\doi}[1]{doi: #1}\else
  \providecommand{\doi}{doi: \begingroup \urlstyle{rm}\Url}\fi

\bibitem[Murray and Escola(2020)]{Murray2020}
J.~M. Murray and G.~S. Escola.
\newblock {Remembrance of things practiced with fast and slow learning in
  cortical and subcortical pathways}.
\newblock \emph{Nature Communications}, 11:\penalty0 6441, 2020.

\bibitem[Eichenbaum(2014)]{Eichenbaum2014}
H.~Eichenbaum.
\newblock {Time cells in the hippocampus: a new dimension for mapping
  memories}.
\newblock \emph{Nature Reviews Neuroscience}, 15\penalty0 (11):\penalty0
  732--744, 2014.

\bibitem[MacDonald et~al.(2011)MacDonald, Lepage, Eden, and
  Eichenbaum]{MacDonald2011}
C.~J. MacDonald, K.~Q. Lepage, U.~T. Eden, and H.~Eichenbaum.
\newblock {Hippocampal ``time cells'' bridge the gap in memory for
  discontiguous events}.
\newblock \emph{Neuron}, 71\penalty0 (4):\penalty0 737--749, 2011.

\bibitem[Karuvally et~al.(2026)Karuvally, Lertsaroj, Sejnowski, and
  Siegelmann]{Karuvally2025}
A.~Karuvally, P.~Lertsaroj, T.~Sejnowski, and H.~Siegelmann.
\newblock Exponential dynamic energy network for high capacity sequence memory.
\newblock \emph{Advances in Neural Information Processing Systems},
  38:\penalty0 39306--39333, 2026.

\bibitem[Betteti et~al.(2026)Betteti, Baggio, and Zampieri]{Betteti2026}
S.~Betteti, G.~Baggio, and S.~Zampieri.
\newblock {A Dynamical Theory of Sequential Retrieval in Input-Driven Hopfield
  Networks}.
\newblock \emph{arXiv preprint arXiv:2603.03201}, 2026.

\bibitem[Chaudhry et~al.(2023)Chaudhry, Zavatone-Veth, Krotov, and
  Pehlevan]{Chaudhry2023}
H.~T. Chaudhry, J.~A. Zavatone-Veth, D.~Krotov, and C.~Pehlevan.
\newblock {Long Sequence Hopfield Memory}.
\newblock In \emph{Advances in Neural Information Processing Systems
  (NeurIPS)}, 2023.

\bibitem[Sompolinsky and Kanter(1986)]{Sompolinsky1986}
H.~Sompolinsky and I.~Kanter.
\newblock {Temporal association in asymmetric neural networks}.
\newblock \emph{Physical Review Letters}, 57\penalty0 (22):\penalty0 2861,
  1986.

\bibitem[Kleinfeld(1986)]{Kleinfeld1986}
D.~Kleinfeld.
\newblock {Sequential state generation by model neural networks}.
\newblock \emph{Proceedings of the National Academy of Sciences}, 83\penalty0
  (24):\penalty0 9469--9473, 1986.

\bibitem[Krotov and Hopfield(2016)]{Krotov2016}
D.~Krotov and J.~J. Hopfield.
\newblock {Dense Associative Memory for Pattern Recognition}.
\newblock In \emph{Advances in Neural Information Processing Systems
  (NeurIPS)}, 2016.

\bibitem[Demircigil et~al.(2017)Demircigil, Heusel, L{\"o}we, Upgang, and
  Vermet]{Demircigil2017}
M.~Demircigil, J.~Heusel, M.~L{\"o}we, S.~Upgang, and F.~Vermet.
\newblock {On a model of associative memory with huge storage capacity}.
\newblock \emph{Journal of Statistical Physics}, 168\penalty0 (2):\penalty0
  288--299, 2017.

\bibitem[Ramsauer et~al.(2021)]{Ramsauer2021}
H.~Ramsauer et~al.
\newblock {Hopfield Networks is All You Need}.
\newblock In \emph{International Conference on Learning Representations
  (ICLR)}, 2021.

\bibitem[Rabinovich et~al.(2008)Rabinovich, Huerta, Varona, and
  Afraimovich]{Rabinovich2008}
M.~I. Rabinovich, R.~Huerta, P.~Varona, and V.~S. Afraimovich.
\newblock {Transient cognitive dynamics, metastability, and decision making}.
\newblock \emph{PLoS Computational Biology}, 4\penalty0 (5):\penalty0 e1000072,
  2008.

\bibitem[Russo and Treves(2012)]{Russo2012}
E.~Russo and A.~Treves.
\newblock {Cortical free-association dynamics: Distinct phases of a latching
  network}.
\newblock \emph{Physical Review E}, 85:\penalty0 051920, 2012.

\bibitem[Recanatesi et~al.(2015)Recanatesi, Katkov, Romani, and
  Tsodyks]{Recanatesi2015}
S.~Recanatesi, M.~Katkov, S.~Romani, and M.~Tsodyks.
\newblock {Neural network model of memory retrieval}.
\newblock \emph{Frontiers in Computational Neuroscience}, 9:\penalty0 149,
  2015.

\bibitem[Krupa and Melbourne(1995)]{Krupa1995}
M.~Krupa and I.~Melbourne.
\newblock {Asymptotic stability of heteroclinic cycles in systems with
  symmetry}.
\newblock \emph{Ergodic Theory and Dynamical Systems}, 15\penalty0
  (1):\penalty0 121--147, 1995.

\bibitem[Field(2020)]{Field1996}
M.~J. Field.
\newblock \emph{{Lectures on Bifurcations, Dynamics and Symmetry}}.
\newblock Pitman Research Notes in Mathematics 356. Chapman and Hall/CRC, 2020.

\bibitem[Ashwin and Postlethwaite(2013)]{Ashwin2013}
P.~Ashwin and C.~Postlethwaite.
\newblock {On designing heteroclinic networks from graphs}.
\newblock \emph{Physica D: Nonlinear Phenomena}, 265:\penalty0 26--39, 2013.

\bibitem[D{\"o}rfler and Bullo(2014)]{Dorfler2014}
F.~D{\"o}rfler and F.~Bullo.
\newblock {Synchronization in complex networks of phase oscillators: A survey}.
\newblock \emph{Automatica}, 50\penalty0 (6):\penalty0 1539--1564, 2014.

\bibitem[Menara et~al.(2022)Menara, Baggio, Bassett, and
  Pasqualetti]{Menara2022}
T.~Menara, G.~Baggio, D.~S. Bassett, and F.~Pasqualetti.
\newblock {Functional control of oscillator networks}.
\newblock \emph{Nature Communications}, 13\penalty0 (1):\penalty0 4721, 2022.

\bibitem[Ogranovich et~al.(2026)Ogranovich, Guo, Venkatakrishnan, Shapiro,
  Bullo, and Pasqualetti]{Ogranovich2026}
A.~Ogranovich, T.~Guo, A.~R. Venkatakrishnan, M.~Shapiro, F.~Bullo, and
  F.~Pasqualetti.
\newblock {Oscillator-Based Associative Memory with Exponential Capacity:
  Theory, Algorithms, and Hardware Implementation}.
\newblock \emph{arXiv preprint arXiv:2604.01469}, 2026.

\end{thebibliography}
\appendix
\subsection*{Proof of Theorem \ref{thm: general graph}}

\begin{proof}
 We first show that the reachable set $\mathcal R^{k+1}$ is convex and that the
energy $E^{k+1}$ is strictly convex on $\mathcal R^{k+1}$. Consequently $\boldsymbol{\delta}^{k+1}$ is the unique equilibrium
(the energy minimum) in $\mathcal R^{k+1}$. We then show that, starting from $\boldsymbol{\delta}^k$, the
energy stays below the boundary value $E^{k+1}_\partial$, so the flow can never escape $\mathcal R^{k+1}$. Trapped
inside $\mathcal R^{k+1}$ with a single attractor, the state must converge to $\boldsymbol{\delta}^{k+1}$.  We make
this proof in four steps.

To begin with, we introduce two facts of the system. When $\omega=0$,
\eqref{eq: system dynamics} is the gradient flow $\dot{\boldsymbol{\theta}}=-\nabla_{\boldsymbol{\theta}} E^{k+1}$, so along any
trajectory
\[
  \tfrac{d}{dt}E^{k+1}(\boldsymbol{\theta}(t))=\nabla_{\boldsymbol{\theta}} E^{k+1}\cdot\dot{\boldsymbol{\theta}}=-\|\nabla_{\boldsymbol{\theta}} E^{k+1}\|^2\le0 ;
\]
thus $E^{k+1}$ is nonincreasing, and strictly decreasing wherever $\nabla_{\boldsymbol{\theta}} E^{k+1}\neq0$, i.e.\ away from
equilibria. Second, since $\boldsymbol{\delta}=B^\top\boldsymbol{\theta}$, the edge differences evolve as
$\dot{\boldsymbol{\delta}}=B^\top\dot{\boldsymbol{\theta}}\in\im B^\top$, so the flow can move $\boldsymbol{\delta}$
only within the cut space $\im B^\top$. We use the first fact in Steps~2--4 and the second in Step~4.

\emph{Step 1: $\mathcal R^{k+1}$ is convex.} The set $\boldsymbol{\delta}^{k+1}+\im B^\top$ is affine, hence
convex, and the box $(-\tfrac\pi2,\tfrac\pi2)^N$ is convex; since the intersection of convex sets is
convex, $\mathcal R^{k+1}=(\boldsymbol{\delta}^{k+1}+\im B^\top)\cap(-\tfrac\pi2,\tfrac\pi2)^N$ is convex. 

\emph{Step 2: $E^{k+1}$ is strictly convex on $\mathcal R^{k+1}$.} The  gradient of $E(\boldsymbol{\delta})$ is
$\nabla E^{k+1}(\boldsymbol{\delta})=\diag(\boldsymbol{a}^{k+1})\sin\boldsymbol{\delta}$, so its Hessian is
$\nabla^2 E^{k+1}(\boldsymbol{\delta})=\diag(a^{k+1}_e\cos\delta_e)$, which is positive definite on the box
because $a^{k+1}_e>0$ and $\cos\delta_e>0$ for $\delta_e\in(-\tfrac\pi2,\tfrac\pi2)$. Hence $E^{k+1}$ is strictly
convex on the box, and in particular on  $\mathcal R^{k+1}$.

\emph{Step 3: $\boldsymbol{\delta}^{k+1}$ is the unique equilibrium in $\mathcal R^{k+1}$.} Strict convexity
(Step~2) already guarantees that $E^{k+1}$ has \emph{at most} one minimizer on the convex set $\mathcal R^{k+1}$. 
It remains to show  that $\boldsymbol{\delta}^{k+1}$ is a minimizer. Recall from Step~2 that
$\nabla E^{k+1}(\boldsymbol{\delta})=\diag(\boldsymbol{a}^{k+1})\sin\boldsymbol{\delta}$. Since $\mathcal R^{k+1}$ lies in the affine set $\boldsymbol{\delta}^{k+1}+\im B^\top$, its tangent directions are exactly the direction space
$\im B^\top=\{B^\top u: u\in\mathbb{R}^n\}$, so every direction tangent to $\mathcal R^{k+1}$ has the form
$v=B^\top u$, and by condition~(i),
\[
  \nabla E^{k+1}(\boldsymbol{\delta}^{k+1})^\top v
  =\bigl(\diag(\boldsymbol{a}^{k+1})\sin\boldsymbol{\delta}^{k+1}\bigr)^\top B^\top u
  =\bigl(B\diag(\boldsymbol{a}^{k+1})\sin\boldsymbol{\delta}^{k+1}\bigr)^\top u
  =0 .
\]

the box because $E^{k+1}$ is nonincreasing, so
$E^{k+1}(\boldsymbol{\delta}(t))\le E^{k+1}(\boldsymbol{\delta}^k)<E^{k+1}_\partial$ by condition~(iii), whereas
$E^{k+1}\ge E^{k+1}_\partial$ on $\partial\mathcal R^{k+1}$, so the boundary is never reached. 

To summarize, Steps~1 and~2 show that $\mathcal R^{k+1}$ is convex and $E^{k+1}$ is strictly convex on it, so by
Step~3 the target $\boldsymbol{\delta}^{k+1}$ is the unique equilibrium in $\mathcal R^{k+1}$. Step~4 then uses the
energy barrier of condition~(iii) to trap the trajectory inside $\mathcal R^{k+1}$, where it can only slide
down to $\boldsymbol{\delta}^{k+1}$. Equivalently, $\boldsymbol{\delta}^k$ lies in the basin of attraction of
$\boldsymbol{\delta}^{k+1}$, so the weight switch retrieves $\boldsymbol{\delta}^{k+1}$ with no spurious
equilibrium in between.
\end{proof}

\subsection*{Proof of Lemma~\ref{lemma: weights mapping}}

\begin{proof}

For the designed weights $\boldsymbol{a}$ to drive the network to the target memory pattern $\boldsymbol{\delta}$, they must satisfy $\diag(\boldsymbol{a}) \sin (\boldsymbol{\delta}) \in \ker B$. We therefore begin by analyzing the incidence matrix $B$ of the honeycomb graph from \cite{Ogranovich2026} and its kernel.  The honeycomb graph $\mathcal{G}_m^{n_c}$ contains $\mathcal{C}_1, \cdots, \mathcal{C}_m$ cycles and according to the  specific node numbering choice in \cite{Ogranovich2026}, the last node of cycle $\mathcal{C}_p$ is the starting node of cycle $\mathcal{C}_{p+1}$. Consecutive cycles therefore overlap in a single row of $B \in \mathbb{R}^{n \times N}$. Orienting every edge along its loop and ordering the edges cycle by cycle, $B$ takes the block-staircase form
\begin{equation}\label{eq:B-honeycomb}
  \scalebox{0.8}{$\displaystyle
  B = \left[\begin{array}{cccc|cccc|c|cccc}
    \multicolumn{4}{c}{\overbrace{\hphantom{-1\;\;-1\;\;-1\;\;-1}}^{n_c\ \text{columns}}} &
    \multicolumn{4}{c}{\overbrace{\hphantom{-1\;\;-1\;\;-1\;\;-1}}^{n_c\ \text{columns}}} &
    \multicolumn{1}{c}{} &
    \multicolumn{4}{c}{\overbrace{\hphantom{-1\;\;-1\;\;-1\;\;-1}}^{n_c\ \text{columns}}} \\[2pt]
    -1 &        &        &  1 &        &        &        &        &                  &        &        &        &    \\
     1 & -1     &        &    &        &        &        &        &                  &        &        &        &    \\
       & \ddots & \ddots &    &        &        &        &        &                  &        &        &        &    \\
       &        &  1     & -1 & -1     &        &        &  1     &                  &        &        &        &    \\
       &        &        &    &  1     & -1     &        &        &                  &        &        &        &    \\
       &        &        &    &        & \ddots & \ddots &        &                  &        &        &        &    \\
       &        &        &    &        &        &  1     & -1     &                  &        &        &        &    \\
       &        &        &    &        &        &        &        & \quad\ddots\quad &        &        &        &    \\
       &        &        &    &        &        &        &        &                  & -1     &        &        &  1 \\
       &        &        &    &        &        &        &        &                  &  1     & -1     &        &    \\
       &        &        &    &        &        &        &        &                  &        & \ddots & \ddots &    \\
       &        &        &    &        &        &        &        &                  &        &        &  1     & -1 \\[4pt]
    \multicolumn{4}{c}{\underbrace{\hphantom{-1\;\;-1\;\;-1\;\;-1}}_{\textstyle B_1}} &
    \multicolumn{4}{c}{\underbrace{\hphantom{-1\;\;-1\;\;-1\;\;-1}}_{\textstyle B_2}} &
    \multicolumn{1}{c}{\cdots} &
    \multicolumn{4}{c}{\underbrace{\hphantom{-1\;\;-1\;\;-1\;\;-1}}_{\textstyle B_m}}
  \end{array}\right]$},
\end{equation}
where each block $B_\ell$ collects the $n_c$ edges of cycle $\mathcal{C}_\ell$, so $B_\ell$ is nonzero only in the rows of that cycle's nodes. When two cycles share a node, their blocks both have a nonzero entry in that node's row. For example, in the honeycomb $\mathcal{G}^5_3$ of Figure~\ref{fig:honeycomb_encoding}, cycles~$1$ and~$2$ share node~$5$, so $B_1$ and $B_2$ are both nonzero in row~$5$. This single shared row is why $B$ is not block diagonal.


We first analyze the kernel of the incidence matrix $B$. Following \eqref{eq:B-honeycomb}, $B$ decomposes cycle by cycle as $B=[\,B_1\ \cdots\ B_m\,]$ with $B_\ell\in\mathbb{R}^{n\times n_c}$. Each block $B_\ell$ is the incidence matrix of a single cycle, on which every node has exactly one incoming and one outgoing edge; hence every row of $B_\ell$ sums to $0$, that is,
\begin{equation}
  B_\ell\,\mathbf{1}_{n_c}=\mathbf{0},\qquad \ell=1,\dots,m .
\end{equation}
For each cycle $\mathcal{C}_\ell$, let $z_\ell\in\mathbb{R}^N$ be the indicator of its edges, with entries equal to $1$ on the edges of $\mathcal{C}_\ell$ and $0$ elsewhere. Stacking these vectors columnwise gives
\begin{equation}\label{eq:kerB-basis}
  Z=[\,z_1\ \ z_2\ \ \cdots\ \ z_m\,]
   =\begin{bmatrix}\mathbf{1}_{n_c} & & \\ & \ddots & \\ & & \mathbf{1}_{n_c}\end{bmatrix}\in\mathbb{R}^{N\times m}.
\end{equation}
Since $B\,z_\ell=B_\ell\,\mathbf{1}_{n_c}=\mathbf{0}$, we have $BZ=\mathbf{0}$ and $\operatorname{im}Z\subseteq\ker B$. The columns $z_1,\dots,z_m$ have pairwise disjoint supports, so they are linearly independent and $\dim(\operatorname{im}Z)=m$. Since $\mathcal{G}_m^{n_c}$ is connected, $\operatorname{rank}B=n-1$, giving
\begin{equation}
  \dim\ker B=N-(n-1)=mn_c-m(n_c-1)=m .
\end{equation} 
 So $\operatorname{im}Z=\ker B$, which means that the columns of $Z$ form a basis of $\ker B$. It remains to show that the weights \eqref{eq: weights mapping} make $\boldsymbol{\delta}^{k+1}$ a stored pattern. By condition~(i) of Theorem~\ref{thm: general graph}, this holds precisely when the weights $\boldsymbol{a}^{k+1}$ satisfy
\begin{equation}\label{eq: condition}
  \diag(\boldsymbol{a}^{k+1})\sin(\boldsymbol{\delta}^{k+1}) \in \ker B .
\end{equation}
With $a_{ij}^{k+1}=1/\lvert\sin\delta_{ij}^{k+1}\rvert$ from \eqref{eq: weights mapping}, we have $a_{ij}^{k+1}\sin(\delta_{ij}^{k+1})=\operatorname{sign}(\delta_{ij}^{k+1})$. By condition~(iii) of Lemma~\ref{lemma: weights mapping}, the phase differences around each cycle $\mathcal{C}_\ell$ share a common sign $s_\ell\in\{\pm1\}$, so we can rewrite \eqref{eq: condition} in terms of the columns of $Z$ as follows:
\begin{equation}
  \diag(\boldsymbol{a}^{k+1})\sin(\boldsymbol{\delta}^{k+1})
  = s_1 z_1 + s_2 z_2 + \cdots + s_m z_m
  = Z\boldsymbol{s}\ \in\ \ker B,
\end{equation}
with $\boldsymbol{s}=(s_1,\dots,s_m)^\top$. Therefore $\boldsymbol{\delta}^{k+1}$ is a stable phase-locked configuration, which completes the argument.

\end{proof}

\subsection*{Proof of Lemma~\ref{lem:volume}}

\begin{proof}
Since $P_{-w}=-P_w$, the two sets have the same volume, and we may take $w>0$. Then condition~(iii) of Lemma~\ref{lemma: weights mapping} makes every phase difference positive, and
\[
P_w=\Bigl\{x\in\bigl(0,\tfrac\pi2\bigr)^{n_c}:\ x_1+\dots+x_{n_c}=2\pi w\Bigr\}.
\]
The proof has three steps.

\emph{Step 1: rescaling.} The substitution $x=\tfrac\pi2u$ maps $P_w$ onto the slice
\[
S_t=\bigl\{u\in(0,1)^{n_c}:\ u_1+\dots+u_{n_c}=t\bigr\},\qquad t=4w,
\]
of the unit cube. Lengths are multiplied by $\pi/2$, so $(n_c-1)$-dimensional volumes are multiplied by $(\pi/2)^{n_c-1}$, and $V_w=(\pi/2)^{n_c-1}\operatorname{vol}(S_{4w})$.

\emph{Step 2: the slice as a probability density.} Let $U$ be uniform on the unit cube $(0,1)^{n_c}$, and let $f_{n_c}$ be the density of the sum $T=U_1+\dots+U_{n_c}$. The probability that $T$ lies in $[t,t+dt]$ is the volume of the part of the cube between the parallel hyperplanes $\sum_iu_i=t$ and $\sum_iu_i=t+dt$. These hyperplanes have unit normal $(1,\dots,1)/\sqrt{n_c}$, so they are a distance $dt/\sqrt{n_c}$ apart, and the volume of the slab is $\operatorname{vol}(S_t)\,dt/\sqrt{n_c}$. Hence $f_{n_c}(t)=\operatorname{vol}(S_t)/\sqrt{n_c}$, that is, $\operatorname{vol}(S_t)=\sqrt{n_c}\,f_{n_c}(t)$.

\emph{Step 3: the density of the sum.} We compute the distribution function $F_{n_c}(t)=\Pr[T\le t]$, which is the volume of $\{u\in(0,1)^{n_c}:\sum_iu_i\le t\}$. Without the upper bounds $u_i<1$, the region $\{u\ge0:\sum_iu_i\le t\}$ is a simplex of volume $t^{n_c}/n_c!$. We remove the points that violate the upper bounds by inclusion--exclusion. For a fixed set of $s$ coordinates, the points of the simplex at which these $s$ coordinates all exceed $1$ form, after subtracting $1$ from each of them, the simplex $\{u\ge0:\sum_iu_i\le t-s\}$, of volume $(t-s)^{n_c}/n_c!$ if $s\le t$ and empty otherwise. There are $\binom{n_c}{s}$ such sets, so
\[
F_{n_c}(t)=\frac{1}{n_c!}\sum_{s=0}^{\lfloor t\rfloor}(-1)^s\binom{n_c}{s}(t-s)^{n_c},
\qquad
f_{n_c}(t)=F_{n_c}'(t)=\frac{1}{(n_c-1)!}\sum_{s=0}^{\lfloor t\rfloor}(-1)^s\binom{n_c}{s}(t-s)^{n_c-1},
\]
which is the Irwin--Hall density. Setting $t=4w$ and combining the three steps gives \eqref{eq:volume}.
\end{proof}

\end{document}